\documentclass[11pt,letterpaper]{article}
\usepackage[T1]{fontenc}
\usepackage[utf8]{inputenc}
\usepackage{lmodern}
\usepackage[margin=0.85in,headheight=14pt,footskip=28pt]{geometry}
\usepackage{amsmath,amssymb,amsthm}
\newtheorem{theorem}{Theorem}
\newtheorem{lemma}[theorem]{Lemma}
\newtheorem{proposition}[theorem]{Proposition}
\theoremstyle{definition}
\newtheorem{algorithm}[theorem]{Algorithm}
\theoremstyle{remark}
\newtheorem{remark}[theorem]{Remark}
\usepackage{microtype}
\usepackage{xcolor}
\usepackage{titlesec}
\usepackage{fancyhdr}
\usepackage{lastpage,verbatim}
\usepackage{hyperref}
\definecolor{linkblue}{RGB}{35,64,100}
\hypersetup{
  colorlinks=true,linkcolor=black,citecolor=black,urlcolor=black,
  pdftitle={Approximating the mixed volumes of a constant number of convex bodies in randomized polynomial time},
  pdfkeywords={mixed volume, membership oracle, randomized approximation, capacity, Gurvits ratio}
}
\newcommand{\R}{\mathbb R}
\newcommand{\Q}{\mathbb Q}
\newcommand{\E}{\mathbb E}
\DeclareMathOperator{\Vol}{Vol}
\DeclareMathOperator{\Capacity}{Cap}
\DeclareMathOperator{\poly}{poly}
\DeclareMathOperator{\interior}{int}
\DeclareMathOperator{\Unif}{Unif}
\newcommand{\op}{\mathrm{op}}
\newcommand{\norm}[1]{\lVert #1\rVert}
\titleformat{\section}{\large\bfseries}{\thesection.}{0.5em}{}
\titleformat{\subsection}{\normalsize\bfseries}{\thesubsection.}{0.5em}{}
\titlespacing*{\section}{0pt}{15pt plus 3pt minus 2pt}{6pt}
\titlespacing*{\subsection}{0pt}{11pt plus 2pt minus 1pt}{4pt}
\title{Approximation algorithm for mixed volumes of a
constant number of convex bodies}
\author{Hariharan Narayanan\thanks{School of Technology and Computer Science, TIFR, Mumbai}}
\date{}

\begin{document}
\maketitle
\vspace{-0.8em}

\begin{abstract}
We study the computation of mixed volumes of convex sets given as
rational subspace sections of bodies described by membership oracles.
The sets need not have nonempty interior. For positive multiplicities
$\alpha_1+\cdots+\alpha_k=n$, we give a randomized algorithm that
computes a relative $\varepsilon$-approximation with probability at least
$1-\delta$. Its running time is polynomial in Gurvits' factor
$\Gamma_\alpha=(n^n/n!)\prod_i\alpha_i!/\alpha_i^{\alpha_i}$,
the input presentation length, and $\varepsilon^{-1}$, and linear in
$1+\log\delta^{-1}$, with degree independent of $k$. The running-time
bound holds for every choice of the random bits, and zero mixed
volumes are detected exactly. We extend the capacity-scaling and
sampling method of Narayanan--Roy by using rational projections and
a smooth mixed Monge--Amp\`ere representation. In particular, the
algorithm runs in randomized polynomial time when the number of
distinct convex sets is bounded above by a constant.
\end{abstract}

\begingroup
\small
\setlength{\parskip}{3pt}
\textbf{AI use.}
OpenAI's Codex was used to develop and prepare this extension of the
capacity-scaling and sampling approach from \cite{NR}, including the smooth estimator,
its finite-bit oracle implementation, capacity optimization, exact
determinant-coefficient computation, complexity and probability
arguments, the worked example, and the mathematical exposition. \cite{NR} was fed as initial input.
\par


\par
\endgroup

\section{Introduction}
\label{sec:introduction}

Let $\widetilde K_1,\ldots,\widetilde K_k\subseteq\R^n$ be compact
convex sets. Their mixed volumes are the coefficients, with the
following normalization, of the volume polynomial
\[
 p(t)=\Vol_n\!\left(\sum_i t_i\widetilde K_i\right)
     =\sum_{|\beta|=n}p_\beta t^\beta,
 \qquad
 p_\alpha=\binom n\alpha
 V(\widetilde K_1[\alpha_1],\ldots,\widetilde K_k[\alpha_k]).
\]
Here $V(K[n])=\Vol_n(K)$, and brackets denote repeated arguments.
 We rescale the summands so that
the coefficient contributes a sufficiently large fraction of the
volume, and express this fraction as the mean of a bounded statistic.

We consider sets $\widetilde K_i=K_i\cap S_i$, where
$B_n\subseteq K_i\subseteq R_0B_n$ is given by a membership oracle
and $S_i$ is a rational linear subspace. An algebraic test decides
whether the mixed volume is zero. On the remaining instances, a
controlled thickening reduces the problem to full-dimensional bodies.

We extend the work of Narayanan--Roy~\cite{NR} is to this oracle model.
We give below an overview of the reductions and the evaluation of
the statistic, followed by the full proofs.

\clearpage
\subsection{Main result}
\label{intro:guarantee}

Let $\alpha\in\mathbb Z_{\ge0}^k$ with $|\alpha|=n$. We discard
the bodies of zero multiplicity, leaving at most $n$ bodies. Define
\begin{equation}\label{eq:gamma-definition}
 \Gamma_\alpha=\frac{n^n}{n!}
       \prod_{i:\alpha_i>0}\frac{\alpha_i!}{\alpha_i^{\alpha_i}}.
\end{equation}
Throughout the paper, $\poly$ denotes a polynomial whose degree and
coefficients are absolute and independent of $k$. We denote the
Euclidean unit ball by $B_n$. Rational bases and numerical inputs
are specified in binary.

\begin{theorem}
\label{thm:main}
Let $K_i$ be convex bodies with $B_n\subseteq K_i\subseteq R_0B_n$,
for a known integer $R_0\ge1$,
and suppose that the $K_i$ have exact membership oracles on rational queries.
Let $S_i$ be rational linear subspaces whose bases have total encoding
length $L_S$. Put $\widetilde K_i=K_i\cap S_i$ and
$V_*=V(\widetilde K_1[\alpha_1],\ldots,\widetilde K_k[\alpha_k])$.
For $0<\varepsilon,\delta<1/2$, there is a randomized algorithm
which returns a nonnegative rational $\widehat V$ such that
\[
 \Pr\{(1-\varepsilon)V_*\le\widehat V\le(1+\varepsilon)V_*\}
 \ge1-\delta.
\]
For every choice of the random bits, the number of original membership
queries and additional bit operations, as well as the bit lengths of
the queries and output, are bounded above by
\begin{equation}\label{eq:main-complexity}
 \poly(n,k,L_S,1+\log R_0,\Gamma_\alpha,\varepsilon^{-1})
          (1+\log\delta^{-1}).
\end{equation}
If $V_*=0$, the algorithm returns zero after a deterministic test.
It terminates on every execution, and the probability of a correct
relative approximation is at least $1-\delta$.
\end{theorem}

The parameter $\Gamma_\alpha$ is the explicit factor in Gurvits' capacity inequality.
Since $1\le\Gamma_\alpha\le n^n/n!<e^n$, the bound for $n$ distinct
arguments is singly exponential in dimension. For fixed active $k$,
Stirling's formula gives $\Gamma_\alpha=O_k(n^{(k-1)/2})$.
The theorem therefore gives a randomized polynomial-time algorithm
for a fixed number of distinct bodies.

We also obtain a polynomial bound when one multiplicity is large.
Suppose the largest multiplicity is $n-r$, and put $a_m=m^m/m!$.
Since $a_m\ge1$ and
$a_{m+1}/a_m=(1+1/m)^m\le e$, one obtains
\[
 \Gamma_\alpha=\frac{a_n}{\prod_i a_{\alpha_i}}
       \le\frac{a_n}{a_{n-r}}\le e^r.
\]
Thus the algorithm is polynomial time for $r=O(\log n)$. In
particular, it applies in the regime of Dyer--Gritzmann--Hufnagel~\cite{DGH}, where
$r\le\psi(n)$ and $\psi(n)\log\psi(n)=o(\log n)$.
The normalization of the input bodies and the uniform estimates are
given in Remark~\ref{rem:polynomial-regimes}. For unrestricted
multiplicities, the complexity bound includes $\Gamma_\alpha$.

\clearpage
\subsection{Oracle access and prior work}
\label{intro:model-comparison}

The membership oracle decides whether a rational point belongs to
$K_i$, with boundary points included. We assign unit cost to each
oracle call and count separately the bit operations used to form and
process the query. When the oracle is implemented in polynomial time,
these bounds also give a polynomial bound for the total computation.
A description of $S_i$ by rational homogeneous equations can be
converted to a rational basis by Gaussian elimination.

The inclusion $B_n\cap S_i\subseteq\widetilde K_i$ provides a known
relative inball. We use it both to test positivity and to bound the
error under thickening. The aspect ratio enters the complexity through
$\log R_0$. Although some subsequent transformations produce small
inradii, the weak-oracle reductions depend polynomially on the lengths
of their descriptions. Section~\ref{sec:model} gives the precise
model and the reductions between membership, separation, and
optimization used here~\cite{GLS,DGH}.

Narayanan--Roy~\cite{NR} give a relative approximation algorithm for
lattice polytopes specified as convex hulls of at most $m_0$ integer
points in $B_\infty(2^L)$. Their running time is
polynomial in $n,m_0,L,\varepsilon^{-1},\log\delta^{-1}$ and
\[
 \widetilde A=\prod_{i=2}^k
       \frac{(\alpha_i+1)^{\alpha_i+1}}{\alpha_i^{\alpha_i}}.
\]
They prove a stronger bound in terms of a refined capacity parameter.
Their result allows the number of bodies to vary. Moreover, Stirling's
bounds give
\[
 \widetilde A\le2n\Gamma_\alpha^2.
\]
Consequently their theorem already gives polynomial dependence on our
ratio for vertex-described inputs, and polynomial time for fixed-$k$.
We extend the method to membership-oracle bodies and their rational
subspace sections.

Their algorithm rescales the polytopes, samples from their Minkowski
sum, and finds the cell of a mixed subdivision containing the sample.
The indicator that the summand faces have the prescribed dimensions
has the required expectation. We construct a smooth weight taking
values in $[0,1]$, expressed in terms of the Hessians of smoothed
support functions. It has the corresponding coefficient fraction as
its mean and can be evaluated by support optimization and rational
matrix computations.

The proof has two parts. A change of variables under the gradient map
identifies the integral of the weight with the target coefficient.
We then bound the cost of evaluating the weight from membership
queries. For this we remove a thin boundary layer, obtain quantitative
curvature estimates, and control the errors in sampling and in the
matrix computations. Exact determinant arithmetic introduces no
further approximation error.

The two local computations have related convex dual problems.
Subgradients of support functions describe the faces in a subdivision;
in our construction the corresponding optimality condition is a
gradient equation. The parameter comparison is given in
Appendix~\ref{sec:nr-comparison}, and the dual problems in
Section~\ref{intro:kkt}. We use their expectation identities, without
requiring the two statistics to agree pointwise.

\clearpage
\subsection{Rational sections}
\label{intro:sections-overview}

We first decide whether the mixed volume is zero. This must be done
before thickening, since a relative approximation to zero is required
to be zero. Clearing denominators in the given bases produces integer
matrices $U_i$ with image $S_i$. Let $H\ge1$ bound the absolute values
of their entries, and consider
\[
 c_\alpha=[z^\alpha]\det\!\left(\sum_i z_iU_iU_i^{\mathsf T}\right).
\]
By the Cauchy--Binet formula, this coefficient is a sum of squared
determinants of selections having $\alpha_i$ basis columns from group
$i$. It is positive if and only if one such selection is independent.
In that case, suitably scaled segments along the selected columns
lie in the relative inballs and have positive mixed volume.
Conversely, enclose each section in a zonotope generated by its basis.
If every selection is dependent, multilinearity shows that the mixed
volume of these zonotopes vanishes. We obtain
\[
 c_\alpha=0\ \Longleftrightarrow\ V_*=0,
 \qquad c_\alpha>0\ \Longrightarrow\
 V_*\ge v_*:=\frac{2^n}{n!(nH)^n}.
\]
For the lower bound we use the fact that a nonzero integer determinant
has absolute value at least one. Finite differences compute $c_\alpha$
exactly in time polynomial in $n,L_S,\Gamma_\alpha$; this step makes
no membership queries.

Suppose now that the mixed volume is positive. Let $P_i$ be the rational
orthogonal projector onto $S_i$, put $P_i^\perp=I-P_i$, and choose
$\eta=\varepsilon v_*/[8n2^n(R_0+1)^{n-1}]$. Define
\[
 C_i=\widetilde K_i+\eta(B_n\cap S_i^\perp),\qquad
 A_i=\eta^{-1}C_i,\qquad
 B_n\subseteq A_i\subseteq\frac{R_0+1}{\eta}B_n.
\]
For rational $y$, we have $y\in A_i$ if and only if
$\eta P_i y\in K_i$ and $y^{\mathsf T}P_i^\perp y\le1$.
Thus membership in these cylinders can be decided exactly with one
original oracle call and rational arithmetic. This particular
thickening does not require exact membership in a general Minkowski sum.

By multilinearity and monotonicity,
\[
 0\le V(C_1[\alpha_1],\ldots,C_k[\alpha_k])-V_*
 \le n2^n(R_0+1)^{n-1}\eta\le\varepsilon V_*/8.
\]
We apply the full-dimensional algorithm to the $A_i$ with relative
accuracy $\xi=\varepsilon/2$ and multiply the estimate by $\eta^n$.
The resulting relative error is at most $\varepsilon$.
The thickness $\eta$ can be exponentially small, but
$\log\eta^{-1}$ is polynomial in the input length and
$\log\varepsilon^{-1}$. Neither $k$ nor $\alpha$ changes in this
reduction, so $\Gamma_\alpha$ is unchanged as well.

The proof is given in Section~\ref{sec:sections-reduction}. In the
rest of this overview, we denote normalized full-dimensional inputs
such as $A_i$ by $F_i$, and their requested relative accuracy by
$\xi$. All subsequent inball assumptions refer to these auxiliary
inputs.

\clearpage
\subsection{Capacity and rescaling}
\label{intro:scaling-overview}

Let $F_i$ be the normalized full-dimensional bodies obtained above,
and put $p_F(t)=\Vol_n(\sum_i t_iF_i)$. We seek relative accuracy
$\xi=\varepsilon/2$. To obtain a lower bound on the coefficient
fraction, we use Gurvits' inequality
\[
 [t^\alpha]p_F\le\Capacity_\alpha(p_F)
 :=\inf_{t>0}\frac{p_F(t)}{t^\alpha}
 \le\Gamma_\alpha[t^\alpha]p_F.
\]
Choose positive rational $\lambda_i$ giving an objective value within
a factor two of the infimum, and put $A=2\Gamma_\alpha$.
At $\lambda$, the volume is at most $A$ times the target monomial's
contribution. This is the same use of capacity scaling as in
Narayanan--Roy~\cite{NR}. We next give an estimator for the resulting
coefficient fraction using the membership oracles.

The dependence on $k$ in computing the scaling requires some care.
In logarithmic coordinates the objective is the convex function
$\log p_F(e^y)-\langle\alpha,y\rangle$. We eliminate one coordinate
by homogeneity. The inner and outer radius bounds then give a search
region of size polynomial in their logarithmic ratio and the
dimension. Convex optimization with approximate volume evaluations
computes the scaling with an absolute polynomial degree. In
particular, the computation has no grid search exponential in $k$.

Apply a common rational affine rounding to the scaled sum, with linear
map $T$ and center $\zeta$. Choose a decomposition
$\zeta=\sum_i\zeta_i$, with $\zeta_i\in T(\lambda_iF_i)$, and set
$L_i=T(\lambda_iF_i)-\zeta_i$. Then
\[
 0\in L_i,\qquad B_n\subseteq S:=\sum_iL_i\subseteq R_1B_n,
 \qquad R_1=2n(n+1).
\]
Each summand lies in the sum because the other summands contain the
origin. Its support function is consequently $R_1$-Lipschitz, even
though the inradius of $L_i$ may be small. If $b_\alpha$ denotes
the target coefficient after these transformations, translation
invariance and homogeneity give
\[
 b_\alpha=|\det T|\lambda^\alpha[t^\alpha]p_F,
 \qquad\Vol_n(S)\le A b_\alpha.
\]
We add a small ball to each summand, setting
\[
 \tau=\frac{\xi}{4096nkA},\qquad Q_i=L_i+\tau B_n,
 \qquad Q=S+k\tau B_n\subseteq(1+k\tau)S.
\]
By monotonicity, every coefficient increases. The sum of their
increases is $\Vol_n(Q)-\Vol_n(S)=O(\xi)b_\alpha$, which also
bounds the change in the target coefficient $b_\alpha^\tau$.
This estimate uses only the inball of the sum. It does not depend
on the reciprocal of an individual summand's inradius.
Sections~\ref{sec:scaling} and~\ref{sec:balls} give the scaling,
centering, and enlargement estimates.

\clearpage
\subsection{A smooth integral representation}
\label{intro:representation-overview}

To express $b_\alpha^\tau$ as an integral of local quantities, let
$h_i(u)=\sup_{y\in L_i}\langle u,y\rangle$ be the support function.
There are two obstructions to using its Hessian directly. The function
need not be differentiable; moreover, at points of twice
differentiability, homogeneity forces a zero radial eigenvalue.
We therefore smooth the support function and add a term with positive
radial curvature. For $G\sim N(0,I_n)$, define
\[
 g_i(u)=\E h_i(u+G),\qquad
 f_i(u)=g_i(u)+\tau\sqrt{1+\norm u^2},\qquad
 H_i(u)=\nabla^2f_i(u).
\]
The Gaussian convolution $g_i$ is smooth and convex. Differentiating
the square-root term and using the Lipschitz bound for the support
function gives
\[
 H_i(u)\succeq\frac{\tau}{(1+\norm u^2)^{3/2}}I_n,
 \qquad
 0\le f_i(u)-h_{Q_i}(u)\le R_1\sqrt n+\tau.
\]
The second bound is uniform in $u$. Thus the smoothed function has
the same recession function as the support function of the enlarged
body. This will imply an exact volume identity; the smoothing itself
introduces no further error.

For positive $t_i$, put $F_t=\sum_i t_i f_i$ and
$Q_t=\sum_i t_iQ_i$. We claim that
\[
 \nabla F_t:\R^n\longrightarrow\interior Q_t
 \quad\text{is a smooth bijection with a smooth inverse}.
\]
Indeed, the gradient of the Gaussian term belongs to $L_i$, and the
gradient of the square-root term belongs to the interior of $\tau B_n$.
The image is therefore contained in $\interior Q_t$. Conversely, fix
$x\in\interior Q_t$. The function $F_t(u)-\langle x,u\rangle$ tends
to infinity as $\norm u$ tends to infinity, since $x$ has a positive
distance from the supporting hyperplanes of the body. Strict
convexity gives a unique minimizer. At that point the gradient is
$x$, proving surjectivity. Positive definiteness of the Hessian also
gives injectivity and a smooth inverse.

Applying the change of variables $x=\nabla F_t(u)$, we obtain
\[
 \begin{aligned}
 \Vol_n(Q_t)&=\int_{\R^n}\det\!\left(\sum_i t_iH_i(u)\right)\,du,\\
 b_\alpha^\tau&=\int_{\R^n}[t^\alpha]
                 \det\!\left(\sum_i t_iH_i(u)\right)\,du.
 \end{aligned}
\]
The first equality follows from the Jacobian formula. To obtain the
second, apply Cauchy--Binet to Gram factorizations of the Hessians.
Every coefficient of the determinant polynomial is a sum of squares,
and hence is nonnegative and bounded by the determinant at all ones.
The latter is integrable by the first equality. We may therefore
integrate the coefficients separately. This gives the required smooth
mixed Monge--Amp\`ere representation, even when the boundaries of
the original bodies have flat faces or corners. In the next step we
transfer the integral to a bounded body, from which we can sample.
Section~\ref{sec:representation} supplies the details.

\clearpage
\subsection{The estimator}
\label{intro:estimator-overview}

We now sample from $Q$ using the standard algorithms for convex
bodies. Put $F=\sum_i f_i$. Given $x\in\interior Q$, let
$u(x)=(\nabla F)^{-1}(x)$ and consider the weight
\[
 W_\alpha(x)=
 \frac{[z^\alpha]\det\!\left(\sum_i z_iH_i(u(x))\right)}
      {\det\!\left(\sum_i H_i(u(x))\right)}.
\]
Here the $z_i$ are formal variables, and the inverse point is fixed
while the coefficient is extracted. The denominator is the sum of
all coefficients of the numerator's determinant polynomial. Their
nonnegativity implies $0\le W_\alpha\le1$.

Let $X$ be uniform on $Q$. Under the substitution
$x=\nabla F(u)$, the Jacobian cancels the denominator, and hence
\[
 \E W_\alpha(X)=\frac{b_\alpha^\tau}{\Vol_n(Q)}
 \ge\frac{1}{A(1+k\tau)^n}\ge\frac1{2A}.
\]
The lower bound on the mean is the consequence of capacity scaling
needed for relative approximation. A multiplicative concentration
bound for bounded random variables shows that
$O(A\xi^{-2}\log\delta^{-1})$ independent samples suffice in the
ideal setting. For the finite-bit algorithm, we must also estimate
the volume, remove a boundary layer, and bound the errors in the
sampled points and in the evaluation of the weights. These errors
are treated separately below.

The relation to the subdivision method of Narayanan--Roy~\cite{NR}
can also be expressed by convex duality. For vertex-described
polytopes $P_i$, the cell-recovery problem decomposes a sampled point
into $y_i\in\lambda_iP_i$ and maximizes a generic linear functional.
Its dual problem is
\[
 \min_u\left\{\sum_i h_{\lambda_iP_i}(u+a_i)
                         -\langle u,X\rangle\right\}.
\]
The optimality condition identifies the faces through shifted
normal-cone conditions with a common multiplier. Their dimensions
determine the indicator used in their estimator. Evaluating that
indicator therefore requires recovering face data.

The optimization problem in our construction is
\[
 u(X)=\mathop{\rm argmin}_u\{F(u)-\langle u,X\rangle\},
 \qquad\sum_i\nabla f_i(u(X))=X.
\]
The optimality condition is now a smooth gradient equation. The
Hessians give the local volume contribution and determine the
fractional weight through rational matrix computations. Thus the
two algorithms use the same capacity normalization, but evaluate
different local quantities. Our quantity can be recovered through
membership oracles once the inverse point and its Hessians have
been approximated. Section~\ref{sec:estimator} proves the expectation
identity and concentration bound. We next describe the estimates
which make this evaluation polynomial time.

\clearpage
\subsection{Conditioning and numerical evaluation}
\label{intro:numerics-overview}

As a point approaches the boundary of $Q$, its inverse $u(x)$ can
become arbitrarily large, and the lower bound on the Hessian can tend
to zero. To obtain uniform bounds, we restrict the samples to a
contraction of $Q$. Set
\[
 \theta=\frac{\xi}{4096nA},\qquad D=(1-\theta)Q,
 \qquad
 \int_{Q\setminus D}W_\alpha(x)\,dx
       \le n\theta\Vol_n(Q)=O(\xi)b_\alpha.
\]
The displayed error estimate follows from $W_\alpha\le1$ and the
volume of the removed layer. Every point of $D$ is at distance at
least $\theta$ from the boundary, because $B_n\subseteq Q$.
The minimizing property of the inverse point then gives
\[
 \norm{u(x)}\le U:=\frac{R_1\sqrt n+k\tau}{\theta}
       =\poly(n,k,A,\xi^{-1})\qquad(x\in D).
\]
We use $F(0)=\E h_S(G)+k\tau$ here, estimating the support function
of the whole sum so as not to introduce another factor $k$.
We will evaluate at rational points in the larger set
$E=(1-\theta/2)Q$. Their inverse points lie in a rational ball on
which the $H_i$ have a common lower eigenvalue bound $\mu>0$,
with $\mu^{-1}$ polynomial in the same parameters. This gives
both a bounded region for the optimization and quantitative control
of perturbations of the inverse point.

Weak convex optimization gives approximate support values from
membership queries. Differentiating the Gaussian density, we express
the Hessian of the smoothed support function as
\[
 \nabla^2g_i(u)=
 \E\!\left[(GG^{\mathsf T}-I_n)
                  \bigl(h_i(u+G)-h_i(u)\bigr)\right].
\]
Subtracting $h_i(u)$ gives the uniform bound
$|h_i(u+G)-h_i(u)|\le R_1\norm G$, independent of $u$.
The analogous formula for the gradient is used in a fixed number
of strongly convex optimization steps. We truncate and quantize
the Gaussian variables and specify an accuracy for each support
query. Moment estimates and concentration then give rational
derivative approximations with polynomial cost and a prescribed
failure probability. The number of steps and samples is bounded
regardless of whether the resulting approximations are accurate.

The lower eigenvalue bound $\mu$ converts sufficiently small absolute
matrix errors into the relative inequalities
\[
 (1-\omega)H_i(u(x))\preceq\widehat H_i
       \preceq(1+\omega)H_i(u(x)).
\]
Choose $\omega$ to be a sufficiently small constant multiple of
$\xi/n$. Monotonicity of mixed discriminants bounds the numerator
and denominator of the weight within factors $(1\pm\omega)^n$.
Lipschitz bounds for the Hessians give the same control when $x$ is
rounded. The required precision is therefore determined by explicit
parameters, without a lower bound on an unknown local coefficient.
The derivative and conditioning estimates are proved in
Section~\ref{sec:conditioning}; the sampling and combined error
estimates appear in Section~\ref{sec:output}.

\clearpage
\subsection{Coefficient extraction and rational sampling}
\label{intro:arithmetic-overview}

We require a coefficient-extraction algorithm whose complexity is
polynomial in the permitted ratio when $k$ varies. Let
$\widehat H_i$ be fixed rational positive-definite approximations
to the Hessians, and clear denominators to obtain integer matrices
$M_i=q_0\widehat H_i$. Put $\alpha!=\prod_i\alpha_i!$.
The following finite-difference identity holds by homogeneity:
\[
 [z^\alpha]\det\!\left(\sum_i z_iM_i\right)
 =\frac1{\alpha!}\sum_{0\le j\le\alpha}
 (-1)^{n-|j|}\binom\alpha j\det\!\left(\sum_i j_iM_i\right).
\]
The forward differences annihilate every degree-$n$ monomial except
the target monomial. Indeed, a surviving exponent must dominate
$\alpha$ in every coordinate, and equality of total degrees then
forces equality of exponents. We use the same rational matrices at
every node, so the signed cancellation takes place in exact integer
arithmetic and contributes no approximation error.

The number of nodes is $\mathcal N_\alpha=\prod_i(\alpha_i+1)$.
Although this can be exponential in the number of bodies, it satisfies
\[
 \mathcal N_\alpha<8n\Gamma_\alpha^2.
\]
To see this, put $s_m=m!e^m/m^m$. By Stirling's bounds,
$m+1\le s_m^2$ and
$\Gamma_\alpha=(\prod_i s_{\alpha_i})/s_n$. The number of nodes is
therefore at most $s_n^2\Gamma_\alpha^2=O(n\Gamma_\alpha^2)$.
Fraction-free elimination and exact summation and division take
polynomial bit time per node. Division by $\det(\sum_iM_i)$ then
gives the rational weight. The same identity is used in the section
zero test, allowing pivoting when matrices are singular.
The arithmetic bounds are proved in
Section~\ref{sec:determinant-arithmetic}.

We must also account for the discretization of sampled points.
The algorithm returns a rational point $\widehat X$, while the
analysis uses $X\sim\Unif(D)$. We construct these on a common
probability space so that
\[
 \Pr\{\|\widehat X-X\|>\sigma\}\le\rho,
 \qquad \widehat X\in E=(1-\theta/2)Q.
\]
A fine-grid sampler with weak membership gives this coupling.
The union of its accepted cells differs from the target body only
within a controlled boundary layer. The corner of a cell is close
to a continuous point chosen within it. We allow separately for
this boundary error, incomplete mixing, and termination of a capped
rejection loop.

The Hessian lower bound controls the inverse gradient map on $E$,
and the Lipschitz estimate transfers spatial error to relative
matrix error. An inverse-polynomial choice of $\sigma$ thus gives
the same multiplicative estimate for the weight as the Hessian
evaluation. We compare the rational and continuous samples by this
coupling; their laws cannot be close in total variation. The Gaussian
variables used for derivative estimation are also truncated and
quantized. Section~\ref{sec:rational-sampling} proves the required
sampling statement.

\clearpage
\subsection{Outline of the algorithm}
\label{intro:algorithm-overview}

\begin{algorithm}[Procedure for one target mixed volume]
\label{alg:main}
Discard the zero multiplicities and compute $c_\alpha$. If it is
zero, return zero. Otherwise construct the exact membership oracles
for $A_i=\eta^{-1}C_i$ and apply the following procedure at relative
accuracy $\xi=\varepsilon/2$.
\begin{enumerate}
\item In one complete run, compute a rational capacity scaling and
a common affine rounding, then jointly center the summands to obtain
$0\in L_i$ and $B_n\subseteq S=\sum_iL_i\subseteq R_1B_n$.
\item Form $Q_i=L_i+\tau B_n$, the smooth functions $f_i$, and
$Q=\sum_iQ_i$. Estimate $\Vol_n(Q)$ and generate
$N=O(\Gamma_\alpha\xi^{-2})$ independent rational samples near
the uniform law on $D=(1-\theta)Q$.
\item For each sample, approximately invert $\nabla F$, compute
rational Hessians, and extract their normalized determinant
coefficient exactly. Multiply the average weight by
$(1-\theta)^n\widehat{\Vol}_n(Q)$ and undo that run's scaling,
linear map, and multinomial normalization.
\item Repeat the entire run independently
$O(1+\log\delta^{-1})$ times, take the median of its unscaled
outputs, and multiply by $\eta^n$.
\end{enumerate}
\end{algorithm}

We bound separately the errors due to section thickening, ball
enlargement, boundary removal, volume estimation, empirical averaging,
sample discretization, and Hessian evaluation. Relative matrix
inequalities control the last two multiplicatively. A union bound
and multiplicative concentration show that a complete run gives
relative core error at most $\xi$ with probability at least $3/4$.
The median of independent runs has failure probability at most
$\delta$. Since the scaling and rounding data may differ between
runs, each output is unscaled before the median is taken.

We impose polynomial bounds on the number of iterations, query
lengths, and intermediate encodings, and return rational fallback
outputs whenever a computation would exceed a bound. The bounds hold for
every execution; the probability assertion concerns the accuracy
of the output. A complete run for full-dimensional inputs has cost
\[
 \poly(n,k,1+\log R,\Gamma_\alpha,\xi^{-1}),
\]
where $R$ is an outer radius bound for the normalized inputs.
The section reduction bounds $\log R$ polynomially in the original
presentation length. The sample count and coefficient grid require
$O(n\Gamma_\alpha^3\xi^{-2})$ determinant evaluations per run.
Support optimization and inversion require further work, which is
also bounded by a polynomial of absolute degree. Independent
repetition contributes the factor $1+\log\delta^{-1}$.

\paragraph{Organization of the paper.}
We specify the oracle model and polynomial-time consequences in
Section~\ref{sec:model}. Section~\ref{sec:sections-reduction} reduces
rational sections to full-dimensional inputs, and
Section~\ref{sec:scaling} constructs the capacity scaling and rounding.
The enlargement, integral representation, and estimator are treated
in Sections~\ref{sec:balls}--\ref{sec:estimator}.
Section~\ref{sec:conditioning} gives the conditioning and arithmetic
estimates, and Section~\ref{sec:output} proves the sampling and final
error bounds. Appendix~\ref{app:rectangle-example} contains a
two-dimensional example. Further comparisons and the convex-duality
calculation appear in Appendix~\ref{sec:nr-comparison}.

\label{intro:last-page}
\clearpage
\section{Oracle model and polynomial-time cases}
\label{sec:model}

Our model of computation is a randomized oracle Turing machine supplied
with independent fair random bits. For the input oracle
$\operatorname{MEM}_{K_i}$, a query $y\in\Q^n$ asks whether $y\in K_i$;
boundary points are included. We charge one oracle call for the answer
and count the bit operations used to form and process the rational query.
The oracle's internal computation is excluded. Thus a polynomial-time
implementation of membership in the query length gives a polynomial
total running time from the bounds below. The supplied integer $R_0$
has binary length $O(1+\log R_0)$; a noninteger bound can first be rounded
upwards. We use concise binary encodings of the accuracy and failure
parameters. Replacing them by smaller dyadic tolerances within a factor
two preserves the asymptotic bounds. The time to read longer encodings,
if supplied, is counted separately.

We specify each $S_i$ by a rational matrix whose columns form a basis,
allowing the empty matrix for $\{0\}$. A description by homogeneous
rational equations is equivalent, since exact Gaussian elimination
computes a basis in polynomial bit time. Let $L_S$ be the total encoding
length of the subspace data. The inclusions
$B_n\cap S_i\subseteq\widetilde K_i\subseteq R_0B_n$ provide relative
inballs, although the sections may have empty ambient interior. In
Section~\ref{sec:sections-reduction} we replace them by full-dimensional
bodies. Membership in each new body is determined by one query to the
original oracle and rational arithmetic.

For derived bodies we use \emph{weak} oracles. Given a convex body
$K$ and a rational tolerance $\rho>0$, define
\[
 K^{+\rho}=K+\rho B_n,
 \qquad K^{-\rho}=\{x: B(x, \rho) \subseteq K\}; B(x, \rho) = x + \rho B_n.
\]
Weak membership certifies either $y\in K^{+\rho}$ or
$y\notin K^{-\rho}$. Within the intervening layer either valid answer
is permitted. A well-presentation also specifies rational data
$B(b,r)\subseteq K\subseteq RB_n$ with $r>0$. For the input bodies,
$b=0$, $r=1$, and $R=R_0$. These are presentations of the ambient
$K_i$. The exact reduction for sections is performed before applying
any full-dimensional weak-oracle procedure. Exact membership implies
weak membership, but we do not assume exact membership in Minkowski sums.

The reductions between weak membership, separation, and optimization
in~\cite{GLS} apply to well-presented convex bodies. Their costs are
polynomial in the dimension, the presentation and query lengths, and
$\log(1/\rho)$ for $0<\rho<1$.
Dyer--Gritzmann--Hufnagel~\cite[Proposition~3, p.~364]{DGH} construct
well-presentations of positive rational Minkowski sums in polynomial
time. Their Proposition~4 gives an oracle-polynomial affine rounding
whose outer-to-inner radius ratio is at most $n\sqrt{n+1}$.
We use rational descriptions, enlarging this bound by a fixed factor
when slack is needed in the inclusions. Weak linear optimization
in~\cite{GLS} also approximates $h_K(v)$ to any specified absolute error.
We count all calls in these reductions as calls to the original input oracles.

Under translations and invertible rational linear maps, we transport
both the presentation data and the weak-oracle tolerances.
Section~\ref{sec:scaling} gives inner balls for the transformed summands.
Their reciprocal radii can be large; the logarithms of the reciprocal radii and
their encoding lengths have polynomial bounds. We apply the volume and
sampling procedures of~\cite{DFK,DGH,GLS} to the resulting
well-presented bodies, choosing weak-oracle tolerances to meet the
accuracy requirements of those procedures.

\paragraph{Reference laws and finite-bit sampling.}
The uniform and Gaussian random variables used in the analysis have
continuous laws. To implement sampling with finitely many random bits
and rational coordinates, we use the following coupling specification.
For spatial tolerance $\sigma>0$ and distributional tolerance $\xi>0$,
the rational output $\widehat X$ is coupled to
$X\sim\Unif(D)$ with
\[
 \Pr\{\norm{\widehat X-X}>\sigma\}\le\xi.
\]
One can instead compare the output law with a specified rounding of a
continuous reference sampler. Any total variation comparison below is
between continuous reference laws or their common quantizations. It is
never between a rational-valued law and continuous uniform measure.
For Gaussian sampling we use the same convention, with truncation and
quantization.

The truncation, quantization, and derivative tolerances are given in
Lemma~\ref{lem:numerical-evaluation}. The coupling guarantee and the
effect of spatial error on the weight are proved in
Section~\ref{sec:rational-sampling}. In Section~\ref{sec:output} we
combine these errors for one run with constant success probability.
Taking the median of independent complete runs then gives logarithmic
dependence on $\delta^{-1}$.

\begin{remark}[Polynomial time in the Dyer--Gritzmann--Hufnagel and fixed-$k$ regimes]
\label{rem:polynomial-regimes}
Theorem~\ref{thm:main} implies randomized polynomial time in the two
cases below. The dependence on the input presentation length,
$\varepsilon^{-1}$, and $1+\log\delta^{-1}$ is polynomial.
When $S_i=\mathbb R^n$, a supplied inner ball permits the normalization
in that theorem: replace $K_i$ by $(K_i-b_i)/r_i$ for
$B(b_i,r_i)\subseteq K_i$, and restore the factor
$\prod_i r_i^{\alpha_i}$ in the mixed volume at the end.
The normalized $\log R_0$ is polynomial in the presentation length.
For sections, we use the specified origin-centered ambient inballs and
rational subspaces, with polynomial dependence on $L_S$.

\emph{A dominant multiplicity.}
After relabelling, let $\alpha_1=\max_i\alpha_i=n-r$, and remove
zero multiplicities. Define $a_m=m^m/m!$ for $m\ge1$.
The inequalities $a_m\ge1$ and
$a_{m+1}/a_m=(1+1/m)^m\le e$ give
\begin{equation}\label{eq:dominant-multiplicity-ratio}
 \Gamma_\alpha
 =\frac{a_n}{\prod_i a_{\alpha_i}}
 \le\frac{a_n}{a_{n-r}}
 =\prod_{j=n-r}^{n-1}\left(1+\frac1j\right)^j
 \le e^r.
\end{equation}
When $r=0$, the empty product gives $\Gamma_\alpha=1$.
Dyer--Gritzmann--Hufnagel~\cite[Theorems~10--11]{DGH} consider
$r\le\psi(n)$, where $\psi$ is nondecreasing and
$\psi(n)\log\psi(n)=o(\log n)$. In particular,
$\psi(n)=o(\log n)$. Hence
$\Gamma_\alpha\le e^{\psi(n)}=n^{o(1)}$, and
\eqref{eq:main-complexity} is polynomial. More generally,
\eqref{eq:dominant-multiplicity-ratio} gives polynomial time whenever
$r=O(\log n)$, irrespective of the other multiplicities and their number.

\emph{A constant number of distinct bodies.}
Suppose the number of active bodies $k$ is fixed. The two-sided
Stirling bounds imply
\[
 \Gamma_\alpha
 \le (2\pi)^{(k-1)/2}
       \exp\!\left(\sum_{i=1}^k\frac1{12\alpha_i}\right)
       \sqrt{\frac{\prod_{i=1}^k\alpha_i}{n}}
 \le (2\pi)^{(k-1)/2}e^{k/12}n^{(k-1)/2}.
\]
Consequently $\Gamma_\alpha=O_k(n^{(k-1)/2})$ uniformly in the
multiplicities, and Theorem~\ref{thm:main} gives randomized polynomial
time in this case as well. In either case the relative error is at
most $\varepsilon$ with probability at least $1-\delta$.
\end{remark}

\paragraph{Zero multiplicities and multiple targets.}
We remove bodies with $\alpha_i=0$ and relabel. If a single body
remains, volume approximation proves the theorem with $\Gamma_{(n)}=1$
when its subspace is $\mathbb R^n$. If the subspace has dimension less
than $n$, the answer is zero. In the remaining case,
$2\le k\le n$ and all $\alpha_i\ge1$. Before the full-dimensional
construction we test for zero mixed volume as in
Section~\ref{sec:sections-reduction}.
If $m$ targets are specified, apply the algorithm to each with failure
probability $\delta/m$. Summing the costs gives the cost for the list,
and a union bound gives joint success probability at least $1-\delta$.

\section{Sections in rational subspaces}
\label{sec:sections-reduction}

We first reduce relative approximation for sections to the
full-dimensional case. Suppose
\[
 \widetilde K_i=K_i\cap S_i,
 \qquad B_n\subseteq K_i\subseteq R_0B_n,
\]
where the rational linear subspace $S_i\subseteq\R^n$ is given by a
rational basis. If instead it is specified by rational equations,
exact Gaussian elimination produces such a basis. Let $L_S$ denote
the total binary encoding length of the subspace descriptions.
Although $\widetilde K_i$ contains $B_n\cap S_i$, this relative unit
ball need not be an ambient inball. Throughout this section mixed
volume is taken in ambient dimension $n$; it may vanish even if
$\sum_iS_i=\R^n$.

\begin{theorem}[Full-dimensional bodies]
\label{thm:full-dimensional}
Let $B_n\subseteq F_i\subseteq RB_n$, where $R\ge1$ is a known
integer and the full-dimensional bodies $F_i$ are given by exact
rational membership oracles. For $\alpha\in\mathbb Z_{\ge0}^k$ with
$|\alpha|=n$, and $0<\xi,\zeta<1/2$, a randomized algorithm produces
a nonnegative rational $\widehat v$ satisfying
\[
 \Pr\!\left[(1-\xi)V(F_1[\alpha_1],\ldots,F_k[\alpha_k])
 \le\widehat v\le
 (1+\xi)V(F_1[\alpha_1],\ldots,F_k[\alpha_k])\right]\ge1-\zeta.
\]
On every execution, the number of membership queries and additional
bit operations, and the query and output bit lengths, are bounded by
\begin{equation}\label{eq:full-dimensional-complexity}
 \poly(n,k,1+\log R,\Gamma_\alpha,\xi^{-1})
       (1+\log\zeta^{-1}),
\end{equation}
with absolute polynomial degree.
\end{theorem}

We prove Theorem~\ref{thm:full-dimensional} in
Sections~\ref{sec:scaling}--\ref{sec:output}. The following deterministic
reduction deduces the result for rational sections from this theorem.

Remove zero multiplicities and clear the denominators of each basis
vector separately. We obtain
\[
 U_i=(u_{i,1},\ldots,u_{i,d_i})\in\mathbb Z^{n\times d_i},
 \qquad \operatorname{im}U_i=S_i,
 \qquad d_i=\dim S_i.
\]
The matrix for the zero subspace is empty. Put
\begin{equation}\label{eq:section-height}
 H=\max\bigl(\{1\}\cup\{|(U_i)_{ab}|:i,a,b\}\bigr),
 \qquad q=nH,
 \qquad v_* =\frac{2^n}{n!q^n}.
\end{equation}
Exact elimination and denominator clearing require polynomial time,
and $1+\log H$ is polynomial in $n,k,L_S$. The basis need not be
orthonormal, so this procedure introduces no irrational coordinates.

\begin{lemma}[Exact zero test and a positive lower bound]
\label{lem:section-zero}
Put
\[
 V_* =V(\widetilde K_1[\alpha_1],\ldots,
             \widetilde K_k[\alpha_k]),
 \qquad
 c_\alpha=[z^\alpha]\det\!\left(\sum_i z_iU_iU_i^{\mathsf T}\right).
\]
We have $V_*=0$ if and only if $c_\alpha=0$. Moreover,
$c_\alpha>0$ implies $V_*\ge v_*$. Exact computation of the integer
$c_\alpha$ takes $\poly(n,k,L_S,\Gamma_\alpha)$ bit operations and
requires no membership queries.
\end{lemma}

\begin{proof}
Concatenate the columns of the $U_i$ to form $U$, labelling each
column by its group. By the Cauchy--Binet formula,
\[
 \det\!\left(\sum_i z_iU_iU_i^{\mathsf T}\right)
  =\sum_{\substack{J\text{ a set of columns of }U\\|J|=n}}
       \det(U_J)^2\prod_{j\in J}z_{\operatorname{group}(j)}.
\]
It follows that $c_\alpha>0$ if and only if there are $n$ linearly
independent columns with $\alpha_i$ columns from group $i$.

Choose such columns $w_1,\ldots,w_n$, ordered according to the
repeated arguments of $V_*$. The norm of $w_j/q$ is at most
$\sqrt nH/(nH)\le1$. Thus $[-w_j/q,w_j/q]$ is contained in the
corresponding $\widetilde K_i$. Since a nonzero integer determinant
has absolute value at least one, monotonicity and the formula for
mixed volumes of segments imply
\[
 V_*\ge
 V([-w_1/q,w_1/q],\ldots,[-w_n/q,w_n/q])
 =\frac{2^n|\det(w_1,\ldots,w_n)|}{n!q^n}
 \ge v_*.
\]

For the converse, enclose each compact set
$\widetilde K_i\subseteq S_i$ in a zonotope
\[
 Z_i=\sum_{a=1}^{d_i}[-T_i u_{i,a},T_i u_{i,a}]
\]
with finite $T_i>0$. If $d_i=0$, take $Z_i=\{0\}$.
By multilinearity, $V(Z_1[\alpha_1],\ldots,Z_k[\alpha_k])$ is a sum
of mixed volumes of $n$ segments. Each term contains $\alpha_i$
generators from group $i$ and is proportional to the absolute value
of their determinant. When $c_\alpha=0$, these determinants all
vanish. A selection with a repeated generator is dependent; for
distinct generators the assertion follows from Cauchy--Binet with
the given group counts. Monotonicity now gives $V_*=0$.

To compute the coefficient, apply the homogeneous degree-$n$
finite-difference identity~\eqref{eq:mixed-finite-difference} to
$U_iU_i^{\mathsf T}$. This identity does not require positive
definiteness. By Lemma~\ref{lem:coefficient-grid}, computing
$c_\alpha$ requires
$\mathcal N_\alpha=\prod_i(\alpha_i+1)<8n\Gamma_\alpha^2$
determinants. For $0\le j_i\le\alpha_i$, the entries of
$\sum_i j_iU_iU_i^{\mathsf T}$ have absolute value at most
$n^2H^2$. Exact determinant computation with pivoting is valid also
for singular matrices and takes bit time polynomial in $n$ and
$1+\log H$. The binomial weights, determinants, and successive
signed sums have polynomial bit lengths; here
$\mathcal N_\alpha\le2^n$. Finally, divide exactly by
$\prod_i\alpha_i!$.
\end{proof}

\subsection{Thickening the sections}

\begin{proof}[Proof of Theorem~\ref{thm:main} from the full-dimensional case]

First apply Lemma~\ref{lem:section-zero}. For $c_\alpha=0$ the
algorithm returns zero. Otherwise $V_*\ge v_*$, and we put
\begin{equation}\label{eq:section-thickness}
 D_0=n2^n(R_0+1)^{n-1},\qquad
 \eta=\frac{\varepsilon v_*}{8D_0}
     =\frac{\varepsilon}
       {8n\,n!q^n(R_0+1)^{n-1}}.
\end{equation}
These are positive rational numbers with $0<\eta<1$. Let $P_i$ denote the
orthogonal projector onto $S_i$, and put
$P_i^\perp=I-P_i$. Rationality follows from
\[
 P_i=U_i(U_i^{\mathsf T}U_i)^{-1}U_i^{\mathsf T}
       \quad(d_i>0),\qquad P_i=0\quad(d_i=0).
\]
Exact rational linear algebra computes these matrices, whose entries
have polynomial bit lengths. Consider the bodies
\begin{equation}\label{eq:section-cylinder}
 C_i=\widetilde K_i+\eta(B_n\cap S_i^\perp)
    =\{x:P_ix\in K_i,\ \norm{P_i^\perp x}\le\eta\}.
\end{equation}
Since $\widetilde K_i$ contains its relative unit ball,
\begin{equation}\label{eq:section-cylinder-bounds}
 \eta B_n\subseteq C_i\subseteq(R_0+1)B_n,
 \qquad
 \widetilde K_i\subseteq C_i
       \subseteq\widetilde K_i+\eta B_n.
\end{equation}
For $A_i=\eta^{-1}C_i$, exact membership at a rational point $y$
is determined by
\begin{equation}\label{eq:section-exact-membership}
 y\in A_i
 \quad\Longleftrightarrow\quad
 \operatorname{MEM}_{K_i}(\eta P_i y)=\mathrm{yes}
       \ \text{and}\ y^{\mathsf T}P_i^\perp y\le1.
\end{equation}
As $P_i^\perp$ is a rational orthogonal projector, the second
comparison is exact. Each query therefore uses at most one original
membership call and polynomially many rational operations. This
constructs the oracle directly, without assuming membership in a
Minkowski sum.

Put $V_\eta=V(C_1[\alpha_1],\ldots,C_k[\alpha_k])$ and
$\kappa_n=\Vol_n(B_n)\le2^n$. Expand the $n$ arguments in the
rightmost inclusion of~\eqref{eq:section-cylinder-bounds} by
multilinearity. A term containing $j$ copies of $\eta B_n$ is
bounded by $\eta^jR_0^{n-j}\kappa_n$, since each remaining argument
is contained in $R_0B_n$. It follows that
\begin{align}
 0\le V_\eta-V_*
 &\le\kappa_n\bigl((R_0+\eta)^n-R_0^n\bigr)\notag\\
 &\le n2^n(R_0+1)^{n-1}\eta
   =\frac{\varepsilon v_*}{8}
   \le\frac{\varepsilon V_*}{8}.
 \label{eq:section-relative-bias}
\end{align}
Individual sections may have zero ambient volume in this argument.
The last inequality uses only positivity of the target mixed volume.

The normalized bodies satisfy $B_n\subseteq A_i\subseteq R_\eta B_n$,
where $R_\eta=\lceil(R_0+1)/\eta\rceil$ is known. Apply
Theorem~\ref{thm:full-dimensional} to the $A_i$ with the same
multiplicities, accuracy $\varepsilon/2$, and failure probability
$\delta$. Denote its output by $\widehat V_A$, and return
$\widehat V=\eta^n\widehat V_A$. Homogeneity and
\eqref{eq:section-relative-bias} give, with probability at least
$1-\delta$,
\[
 (1-\varepsilon/2)V_*
 \le\widehat V
 \le(1+\varepsilon/2)(1+\varepsilon/8)V_*
 \le(1+\varepsilon)V_*.
\]
The lower endpoint exceeds $(1-\varepsilon)V_*$, proving the error bound.

Finally,
\[
 \log\eta^{-1}
 =O\!\left(n\log(n+1)+n\log H
             +n\log(R_0+1)+\log\varepsilon^{-1}\right),
\]
with the same bound for $\log R_\eta$. The thickness and the
normalized presentations consequently have polynomial encoding
length. Substitution into Theorem~\ref{thm:full-dimensional},
together with the costs of the projections and zero test, yields
\eqref{eq:main-complexity}. The number of summand types and the
multiplicities are unchanged, so the ratio remains
$\Gamma_\alpha$. The final multiplication by $\eta^n$ respects the
claimed output-length and bit-complexity bounds.
\end{proof}

\begin{remark}[Zero mixed volume]
\label{rem:section-zero-first}
Thickening can make the target mixed volume positive when it was
originally zero, whereas relative approximation to zero must return
zero exactly. The algebraic test decides this case without querying
the membership oracles. On positive instances, the lower bound from
the rational subspace data controls the error of thickening. A limit
of full-dimensional instances alone would provide neither of these
guarantees, nor the required finite-bit running-time bound.
\end{remark}

\paragraph{Notation in the full-dimensional case.}
Until the end of Section~\ref{sec:output}, let $K_i$ denote the
full-dimensional inputs of Theorem~\ref{thm:full-dimensional} and
let $R_0$ be their known outer radius bound. In the reduction above,
these bodies are $A_i=\eta^{-1}C_i$, the radius bound is $R_\eta$,
and the accuracy parameter is half the accuracy requested for the
sections. We write $\varepsilon,\delta$ for the local accuracy and
failure parameters and set
\[
 p(t)=\Vol_n\!\left(\sum_i t_iK_i\right),\qquad
 p_\alpha=\binom n\alpha V(K_1[\alpha_1],\ldots,K_k[\alpha_k]).
\]
The ambient inballs, positive pure coefficients, and full-dimensional
rounding used below therefore refer to these auxiliary bodies.
We continue to reserve $\widetilde K_i$ for the original sections.

\section{Capacity scaling and the fraction $1/\Gamma_\alpha$}
\label{sec:scaling}

\begin{lemma}[Capacity scaling and joint centering]
\label{lem:scaling}
For positive multiplicities, one can choose a rational scaling within
a factor two of capacity, a common rational rounding, and joint
translations to obtain bodies $L_i$ satisfying $0\in L_i$,
$B_n\subseteq S=\sum_iL_i\subseteq R_1B_n$, and
\[
 \Vol_n(S)\le 2\Gamma_\alpha b_\alpha,
 \qquad b_\alpha=|\det T|\lambda^\alpha p_\alpha.
\]
The scaling, rounding, and translation data have bit length polynomial
in $n,k,1+\log R_0$. Capacity search requires polynomially many volume
calls with inverse-polynomial relative accuracy. The degree of this
polynomial is absolute, and the assertions hold when the volume calls
succeed.
\end{lemma}

We prove the lemma by constructing the scaling and then the common
rounding and translations. We also specify the presentation data
needed for the oracle reductions. The preliminary reduction in the
introduction ensures that all multiplicities are positive.

Define
\[
  \Capacity_\alpha(p)=\inf_{t>0}\frac{p(t)}{t^\alpha},
  \qquad
  \Gamma_\alpha=\frac{n^n}{n!}
  \prod_{i=1}^k\frac{\alpha_i!}{\alpha_i^{\alpha_i}}.
\]
Then
\begin{equation}\label{eq:capacity}
  p_\alpha\le\Capacity_\alpha(p)\le\Gamma_\alpha p_\alpha,
  \qquad \Gamma_\alpha=O_k\!\left(n^{(k-1)/2}\right).
\end{equation}
To deduce this from Gurvits's mixed-volume capacity inequality
\cite[Theorem 2.1]{Gurvits}, replace $t_i$ by the average of
$\alpha_i$ new variables. More precisely, put
\[
 q(z)=p\!\left(\frac1{\alpha_1}\sum_{j=1}^{\alpha_1}z_{1j},
               \ldots,
               \frac1{\alpha_k}\sum_{j=1}^{\alpha_k}z_{kj}\right).
\]
For fixed positive group averages $t_i$, the arithmetic--geometric
mean inequality yields $\prod_{i,j}z_{ij}\le\prod_i t_i^{\alpha_i}$,
with equality for constant groups. The squarefree capacity of $q$
is therefore $\Capacity_\alpha(p)$, and its squarefree coefficient
is $p_\alpha\prod_i\alpha_i!/\alpha_i^{\alpha_i}$. Applying Gurvits'
inequality to $q$ proves the upper bound in~\eqref{eq:capacity}.
Nonnegativity of the coefficients gives the lower bound. For fixed
$k$, the stated growth follows from Stirling's formula.

\subsection{Polynomial-time capacity search}
\label{sec:uniform-capacity}

We minimize the logarithm of the capacity objective by approximate
volume evaluations. The number of evaluations is polynomial, with
degree independent of $k$. After removing zero multiplicities,
$1\le k\le n$ and $\alpha_i\ge1$. For
$p(t)=\Vol_n(\sum_i t_iK_i)$, let
\[
 \Phi(y)=\log p(e^{y_1},\ldots,e^{y_k})-\alpha\cdot y,
 \qquad \Phi_* = \inf_{y\in\R^k}\Phi(y).
\]

\begin{lemma}[Uniform capacity search]
\label{lem:uniform-capacity}
For $0<\zeta_{\rm cap}<1/2$, there is a randomized algorithm that
computes positive rational numbers $\lambda_1,\ldots,\lambda_k$
satisfying
\begin{equation}\label{eq:near-capacity}
 \frac{p(\lambda)}{\lambda^\alpha}
 \le 2\Capacity_\alpha(p).
\end{equation}
The probability that the inequality holds is at least $1-\zeta_{\rm cap}$.
On every execution, including unsuccessful ones, the number of oracle
calls and additional bit operations is bounded by a polynomial in
$n,k,1+\log R_0,\log(1/\zeta_{\rm cap})$ whose degree is independent
of $k$. The output coordinates have polynomial bit length.
\end{lemma}

\begin{proof}
Since the coefficients of $p$ are nonnegative, we may define the
probability distribution
\[
 \Pr_y\{\beta\}=\frac{p_\beta e^{\beta\cdot y}}{p(e^y)},
 \qquad |\beta|=n,
\]
for which differentiation gives
\[
 \nabla\Phi(y)=\mathbb E_y\beta-\alpha,
 \qquad \nabla^2\Phi(y)=\operatorname{Cov}_y(\beta).
\]
These identities imply convexity of $\Phi$, the bounds
$\|\nabla\Phi(y)\|_2\le2n$ and
$0\preceq\nabla^2\Phi(y)\preceq n^2I_k$, and
\begin{equation}\label{eq:capacity-lipschitz-uniform}
 |\Phi(y)-\Phi(z)|\le2n\|y-z\|_\infty.
\end{equation}
Here we used $\|\beta\|_2\le n$ and
$\|\mathbb E_y\beta-\alpha\|_1\le2n$.

We shall restrict the search to a bounded box containing a minimizer.
Homogeneity gives $\Phi(y+c\mathbf1)=\Phi(y)$, so we first impose
$\alpha\cdot y=0$. Let $v_n=\Vol_n(B_n)$ and $L=\log(kR_0)$.
The pure monomials, together with the input inclusions, give
\[
 p(e^y)\ge v_ne^{ny_i}\quad(i=1,\ldots,k),
 \qquad p(\mathbf1)\le v_n(kR_0)^n.
\]
On the sublevel set $\Phi(y)\le\Phi(0)$ in the hyperplane, it follows
that $y_i\le L$ for each $i$. Also,
$\alpha_i y_i=-\sum_{j\ne i}\alpha_jy_j\ge-(n-\alpha_i)L$.
Hence $-nL\le y_i\le L$. The sublevel set is nonempty and compact,
so it contains a minimizer. Define
\[
 \ell=\max\{1,\lceil\log_2(kR_0)\rceil\},\quad
 B=n\ell+1,\quad C=[-B,B]^k,\quad D=2Bk,\quad G=2n+1.
\]
The integer $\ell$ is obtained from the bit length of $kR_0$.
The box $C$ contains a global minimizer $y_*$ and has Euclidean
diameter at most $D$. We use the following tolerances, step size,
and fixed number of iterations:
\[
 e=\frac1{32D},\qquad h=\frac{e}{4n^2k},\qquad
 a=\frac{eh}{8k},\qquad
 s=\frac1{32G^2},\qquad T=(32DG)^2.
\]
For a rational point $v\in[-B-1,B+1]^k$, an approximation to
$\Phi(v)$ within absolute error $a$ requires one relative volume
evaluation of a positive rational Minkowski sum. Choose positive
rational $r_i$ satisfying $|\log r_i-v_i|\le a/(8n)$ and estimate
$p(r)$ within relative error $a/16$. Then approximate
$\log\widehat p(r)-\alpha\cdot v$ rationally. Each of the first two
errors is at most $a/8$: the function $y\mapsto\log p(e^y)$ is
$n$-Lipschitz in $\ell_\infty$, and $|\log(1+u)|\le2|u|$ for
$|u|\le1/2$. Allocate absolute error $a/8$ to the logarithm and
$a/8$ to the final dyadic rounding. The sum is less than $a$.
We place all outputs on one fixed dyadic grid of mesh at most $a/4$.

Since the coordinates $v_i$ have polynomial magnitude, exponential
and logarithm evaluation to these accuracies takes polynomial bit
time. The presentation of $\sum_i r_iK_i$ has polynomial length,
with known inner radius $\sum_i r_i$ and outer radius
$R_0\sum_i r_i$. The weak-oracle reductions and the volume algorithm
therefore apply. To obtain logarithmic dependence on the failure
allowance from an expected polynomial-time volume algorithm, first
use a constant number of repetitions to obtain success probability
at least $7/8$. Cap the resulting run at sixteen times its
expected-cost bound. Markov's inequality reduces the success
probability by at most $1/16$, so it remains greater than $3/4$.
The median of $O(1+\log\rho^{-1})$ independent capped runs fails
with probability at most $\rho$. When a cap is reached, return a
fixed bounded rational value. The cost of a volume call is thus
bounded on every execution and depends logarithmically on
$\rho^{-1}$.
We use these time limits throughout. An illegal output is replaced
by a fixed dyadic value, and every value output is clipped to
$[-M,M]$, with $M=4n(B+n+\ell+2)$. From
$(2/\sqrt n)^n\le v_n\le2^n$ and the radius bounds,
$|\Phi(v)|<M$ on the query box. Clipping therefore preserves the
accuracy of every successful evaluation and bounds the bit lengths
of all outputs, including those on unsuccessful executions.

Starting at $y_0=0$, perform the following update for
$t=0,\ldots,T-1$, using forward differences of the approximate values:
\[
 \widetilde g_{t,i}
 =\frac{\widetilde\Phi(y_t+he_i)-\widetilde\Phi(y_t)}h,
 \qquad
 y_{t+1}=\Pi_C(y_t-s\widetilde g_t).
\]
Projection $\Pi_C$ onto the box is coordinatewise clipping and
uses exact rational arithmetic without square roots. When the
$k+1$ evaluations in an iteration are accurate, the Hessian bound
implies
\[
 |\widetilde g_{t,i}-\partial_i\Phi(y_t)|
 \le \frac{n^2h}{2}+\frac{2a}{h}
 =\frac{3e}{8k},
 \qquad
 \|\widetilde g_t-\nabla\Phi(y_t)\|_2\le e,
 \qquad \|\widetilde g_t\|_2\le G.
\]
By convexity and nonexpansiveness of the projection,
\[
 2s\bigl(\Phi(y_t)-\Phi_*\bigr)
 \le \|y_t-y_*\|_2^2-\|y_{t+1}-y_*\|_2^2
      +s^2G^2+2seD.
\]
Sum over the iterations and apply Jensen's inequality to
$\bar y=T^{-1}\sum_{t=0}^{T-1}y_t$. If every evaluation is accurate,
we obtain
\[
 \Phi(\bar y)-\Phi_*
 \le\frac{D^2}{2sT}+\frac{sG^2}{2}+eD
 =\frac1{64}+\frac1{64}+\frac1{32}
 =\frac1{16}.
\]
Choose positive rational $\lambda_i$ with
$|\log\lambda_i-\bar y_i|\le1/(64n)$. Applying
\eqref{eq:capacity-lipschitz-uniform} once more gives
\[
 \log\frac{p(\lambda)}{\lambda^\alpha}-\Phi_*
 \le\frac1{16}+\frac1{32}=\frac3{32}<\log2,
\]
This proves the factor-two approximation.

The computation makes $T(k+1)$ volume calls. Assign each call
conditional failure probability $\zeta_{\rm cap}/[T(k+1)]$ and
independent fresh random bits. The volume guarantee holds for each
adaptively selected query. A conditional union bound therefore gives
probability at most $\zeta_{\rm cap}$ for any inaccurate evaluation.
The bounds on iterations and bit lengths hold without conditioning
on success: clipping keeps the iterates in $C$, values lie on a fixed
bounded dyadic grid, and $h,s$ are fixed rationals. All increments
consequently lie on a rational lattice with a polynomial-length
denominator. Clipping at the integer endpoints preserves such a
common denominator, and averaging introduces only the factor $T$.
All costs are bounded by the asserted uniform polynomial.
\end{proof}

\subsection{Rounding and joint centering}

Apply the rational affine rounding of~\cite[Proposition~4]{DGH}.
Write the map as $x\mapsto Tx-\zeta$, with $T$ invertible and
rational and $\zeta$ rational. For
$\overline K_i=T(\lambda_iK_i)$ and
$\overline K=\sum_i\overline K_i$, it gives
\[
  B(\zeta,1)\subseteq \overline K\subseteq B(\zeta,R_1),
  \qquad R_1=2n(n+1).
\]
Here $B(\zeta,s)=\zeta+sB_n$. Unlike the input radius $R_0$,
which depends on the geometry, $R_1$ is a fixed rational bound for
the rounded sum, with slack relative to $n\sqrt{n+1}$.
The original individual radius bounds concern $K_i$ before capacity
scaling. As $\zeta\in \overline K$, there are
$\zeta_i\in \overline K_i$ such that $\sum_i \zeta_i=\zeta$.
Define $L_i=\overline K_i-\zeta_i$. These joint translations satisfy
\[
  0\in L_i\quad\text{for every }i,\qquad
  B_n\subseteq S:=\sum_iL_i=\overline K-\zeta\subseteq R_1B_n.
\]
We henceforth use the center $a=0$. Since all summands contain the
origin,
\[
  L_i\subseteq S\subseteq R_1B_n.
\]
The individual $L_i$ need not contain a unit ball.

To compute these translations, we use well-bounded weak convex
feasibility on the slice $\sum_i z_i=\zeta$ as in~\cite{GLS}.
The required finite precision follows from explicit inballs.
Let $\ell=\sum_i\lambda_i$. The inclusions
$\ell B_n\subseteq\sum_i\lambda_iK_i\subseteq R_0\ell B_n$,
compared in width with $B(\zeta,1)\subseteq\overline K$, imply
$\sigma_{\min}(T)\ge1/(R_0\ell)$. Thus we have rational data
\[
 p_i=0,\qquad r_i=\frac{\lambda_i}{R_0\ell},\qquad
 B(p_i,r_i)\subseteq\overline K_i.
\]
The bounds on the logarithmic scaling coordinates give polynomial
bit lengths for the presentation, including $\log(1/r_i)$.
Choose a rational $\gamma$ of polynomial bit length such that
$1/[8(1+R_1)]\le\gamma\le1/[4(1+R_1)]$, and contract the bodies to
\[
  \overline K_i^-=(1-\gamma)\overline K_i+\gamma p_i,
  \qquad \overline K_i^-+\gamma r_iB_n\subseteq \overline K_i,
\]
where membership follows by affine substitution. We claim
that the slice of $\prod_i \overline K_i^-$ contains a relative ball
of radius at least $\gamma\min_i r_i$. Put $p=\sum_i p_i$.
Because $\norm{p-\zeta}\le R_1$, we have
$q=(\zeta-2\gamma p)/(1-2\gamma)$ in
$B(\zeta,1)\subseteq \overline K$. Decompose $q=\sum_i q_i$,
where $q_i\in \overline K_i$. Then
$z_i=(1-2\gamma)q_i+2\gamma p_i$ sum to $\zeta$, and
$B(z_i,\gamma r_i)\subseteq \overline K_i^-$, proving the claim.
Represent the affine constraint exactly in rational coordinates and
solve weak feasibility to error less than $\gamma\min_i r_i/2$
in the ambient product norm. This produces rational translation
vectors in $\overline K_i$ with sum exactly $\zeta$. The inradius
bounds show that the required precision has polynomial bit length.
For the translated body $L_i$, we know the inner ball
$B(-\zeta_i,r_i)$. Merely having $0\in L_i$ does not give an inball
at the origin. We use these presentation data for support optimization
as described in Section~\ref{sec:model}.

For the corresponding volume-polynomial coefficient $b_\alpha$, set
$A=2\Gamma_\alpha$. Combining
\eqref{eq:capacity}--\eqref{eq:near-capacity} gives
\begin{equation}\label{eq:volume-fraction}
  \boxed{\Vol_n(S)\le A b_\alpha.}
\end{equation}
By translation invariance of mixed volumes and the scaling rule,
\[
  b_\alpha=|\det T|\lambda^\alpha p_\alpha.
\]

\section{Addition of small balls}
\label{sec:balls}

Set $C=4096$ and define
\[
  \tau=\frac{\varepsilon}{CnkA},\qquad
  Q_i=L_i+\tau B_n,\qquad Q=\sum_iQ_i=S+k\tau B_n.
\]
Let $b_\alpha^\tau$ be the coefficient for $Q_i$. By monotonicity,
each volume-polynomial coefficient increases. It follows that
\begin{equation}\label{eq:regularization}
\begin{aligned}
  0\le b_\alpha^\tau-b_\alpha
  &\le \Vol_n(S+k\tau B_n)-\Vol_n(S)\\
  &\le\bigl((1+k\tau)^n-1\bigr)\Vol_n(S)
   =O(\varepsilon)b_\alpha.
\end{aligned}
\end{equation}
The second inequality follows from $B_n\subseteq S$. Only the
inradius of the sum enters this estimate; the individual rounded
bodies may have arbitrarily small inradii.

\section{An integral representation by Hessian determinants}\label{sec:representation}

\begin{proposition}[Smooth mixed Monge--Amp\`ere representation]
\label{prop:representation}
The functions $f_i$ defined below satisfy
$H_i=\nabla^2f_i\succ0$. For each $t>0$, the gradient map
$\nabla(\sum_i t_if_i)$ is a diffeomorphism from $\R^n$ onto
$\interior(\sum_i t_iQ_i)$, and
\[
 \Vol_n\!\left(\sum_i t_iQ_i\right)
 =\int_{\R^n}\det\!\left(\sum_i t_iH_i(u)\right)\,du,
 \qquad
 b_\alpha^\tau=\int_{\R^n}[t^\alpha]\det\!\left(\sum_i t_iH_i(u)\right)\,du.
\]
All determinant coefficients are nonnegative and integrable. This is
the smooth case of~\cite[Proposition~3]{PR}; see
Appendix~\ref{sec:nr-comparison}. We give the proof together with the
curvature bounds needed for the algorithm.
\end{proposition}

\subsection{The volume polynomial}

Recall the bodies $Q_i=L_i+\tau B_n$, with
$L_i\subseteq R_1B_n$ and $\tau>0$. For
$t=(t_1,\ldots,t_k)$, where $t_i>0$ for each index, put
\[
 Q_t=\sum_{i=1}^k t_iQ_i,
 \qquad p_\tau(t)=\Vol_n(Q_t),
 \qquad b_\alpha^\tau=[t^\alpha]p_\tau(t).
\]
The notation $[t^\alpha]$ extracts the coefficient of
$t^\alpha=\prod_i t_i^{\alpha_i}$, where $\sum_i\alpha_i=n$.
We shall represent $b_\alpha^\tau$ by a nonnegative integral.
For this, we first write $p_\tau(t)$ as the integral of a determinant
polynomial and identify its coefficients. The representation is exact
for $Q_i$; Section~\ref{sec:balls} bounds the preceding enlargement error.

\subsection{Support functions}

For $L_i$, define the support function
\[
 h_i(u)=\sup_{x\in L_i}x\cdot u,\qquad u\in\R^n.
\]
This function is convex and positively homogeneous. Additivity of
support functions under Minkowski addition gives
\[
 h_{Q_i}(u)=h_i(u)+\tau\norm{u},
 \qquad h_{Q_t}(u)=\sum_i t_i h_{Q_i}(u).
\]
The inclusion $L_i\subseteq R_1B_n$ implies the Lipschitz bound
\[
 |h_i(u)-h_i(v)|\leq R_1\norm{u-v}.
\]

The addition of a ball does not in general make the support function
differentiable. When $u\ne0$, the supporting face of $L_i+\tau B_n$
in direction $u$ is the supporting face of $L_i$ translated by
$\tau u/\norm{u}$. A non-singleton face therefore remains so, and
nondifferentiability persists. Thus smoothing the corners of the
boundary by adding a ball need not remove flat faces or smooth the
support function. At the origin there is a further obstruction:
differentiability and positive homogeneity would force a support
function to be linear, which is impossible for a non-singleton body.

Even at points of twice differentiability, $h_{Q_i}$ cannot give the
gradient map we need. For any support function $h$, homogeneity gives,
where the Hessian exists away from the origin,
\[
 \nabla^2h(u)u=0.
\]
The radial direction is consequently in the kernel of the Hessian.
We therefore smooth the support function and add positive curvature
to obtain an invertible gradient map.

\subsection{Gaussian smoothing and positive curvature}

For $G\sim N(0,I_n)$, set
\[
 g_i(u)=\E h_i(u+G),\qquad
 f_i(u)=g_i(u)+\tau\sqrt{1+\norm{u}^2},\qquad
 H_i(u)=\nabla^2f_i(u).
\]
The function $g_i$ is smooth, since the linear growth of $h_i$
permits differentiation of the Gaussian density under the integral.
Moreover, $g_i$ is an average of translates of the convex function
$h_i$, so $\nabla^2g_i(u)\succeq0$. Gaussian convolution is a
convenient choice here. Adding only the smooth square-root term to
$h_i$ would not remove its singularities.

To bound the curvature from below, put $r=\norm{u}$ and
$\psi(u)=\tau\sqrt{1+\norm{u}^2}$. Direct differentiation gives
\[
 \nabla^2\psi(u)
 =\tau\left(
 \frac{I_n}{\sqrt{1+r^2}}
 -\frac{uu^{\mathsf T}}{(1+r^2)^{3/2}}\right).
\]
In the radial and orthogonal directions, respectively, this gives
\[
 \begin{aligned}
 \nabla^2\psi(u)u
 &=\frac{\tau\bigl((1+r^2)-r^2\bigr)}{(1+r^2)^{3/2}}u
 =\frac{\tau}{(1+r^2)^{3/2}}u,\\
 \nabla^2\psi(u)v
 &=\frac{\tau}{\sqrt{1+r^2}}v
 \qquad\text{when }v\perp u.
 \end{aligned}
\]
For $u\ne0$, the radial direction $u$ has eigenvalue
$\tau(1+r^2)^{-3/2}$, while directions orthogonal to $u$ have
eigenvalue $\tau(1+r^2)^{-1/2}$. At $u=0$ the eigenvalues are all
$\tau$. Since the Hessian of $g_i$ is positive semidefinite,
\begin{equation}\label{eq:curvature}
 H_i(u)\succeq
 \frac{\tau}{(1+\norm{u}^2)^{3/2}}I_n.
\end{equation}
It follows that $f_i$ is smooth and strictly convex. The curvature
bound is positive for each finite $u$, but is not uniform on $\R^n$.

We next bound the difference from the support function. By Jensen's
inequality, $\E G=0$, and the Lipschitz bound,
\[
 0\leq g_i(u)-h_i(u)
 \leq R_1\E\norm{G}
 \leq R_1\sqrt n.
\]
Also,
\[
 0<\sqrt{1+r^2}-r
 =\frac{1}{\sqrt{1+r^2}+r}\leq1.
\]
Using $h_{Q_i}(u)=h_i(u)+\tau\norm{u}$, we obtain
\begin{equation}\label{eq:recession-bound}
 0\leq f_i(u)-h_{Q_i}(u)\leq R_1\sqrt n+\tau.
\end{equation}

Thus $f_i$ has recession function $h_{Q_i}$. Namely, for any base
point $u_0$ and direction $v$,
\[
 \lim_{s\to\infty}
 \frac{f_i(u_0+sv)-f_i(u_0)}{s}=h_{Q_i}(v).
\]
The bounded discrepancy tends to zero after division by $s$,
and homogeneity and continuity imply
$h_{Q_i}(u_0+sv)/s=h_{Q_i}(v+u_0/s)\to h_{Q_i}(v)$.
In particular the directional linear growth is unchanged by smoothing.
Writing
\[
 F_t(u)=\sum_i t_i f_i(u),\qquad M=R_1\sqrt n+\tau,
\]
we have
\[
 h_{Q_t}(u)\leq F_t(u)
 \leq h_{Q_t}(u)+M\sum_i t_i.
\]

\subsection{The image of the gradient map}

We claim that
\[
 \nabla F_t:\R^n\longrightarrow\operatorname{int}Q_t
\]
is a diffeomorphism. By~\eqref{eq:curvature}, its derivative
$\nabla^2F_t(u)=\sum_i t_iH_i(u)$ is positive definite.

First consider the image. The Lipschitz function $h_i$ is
differentiable almost everywhere. At a differentiability point,
$\nabla h_i(z)$ is a supporting point of $L_i$. The Lipschitz bound
allows differentiation under the expectation, giving
\[
 \nabla g_i(u)=\E\nabla h_i(u+G)\in L_i.
\]
The inclusion follows from closedness and convexity of $L_i$.
Hence
\[
 \nabla f_i(u)=\nabla g_i(u)
 +\frac{\tau u}{\sqrt{1+\norm{u}^2}}
 \in L_i+\operatorname{int}(\tau B_n)
 \subseteq\operatorname{int}Q_i.
\]
A positive weighted sum therefore satisfies
$\nabla F_t(u)\in\operatorname{int}Q_t$.

Conversely, let $x\in\operatorname{int}Q_t$. Choose $\rho>0$ such
that $x+\rho B_n\subseteq Q_t$. Then
\[
 F_t(u)-x\cdot u
 \geq h_{Q_t}(u)-x\cdot u
 \geq \rho\norm{u}.
\]
It follows that $F_t(u)-x\cdot u$ tends to infinity as
$\norm{u}\to\infty$. Continuity and strict convexity give a unique
minimizer $u_x$, and its first-order condition is
$\nabla F_t(u_x)=x$. This proves surjectivity.

For injectivity, suppose $u\ne v$ and set $w=u-v$. We have
\[
 \bigl(\nabla F_t(u)-\nabla F_t(v)\bigr)\cdot w
 =\int_0^1 w^{\mathsf T}\nabla^2F_t(v+sw)w\,ds>0.
\]
The gradients at distinct points are therefore distinct.
The inverse function theorem applies everywhere because the Hessian
is positive definite. The resulting smooth local inverses agree by
bijectivity and define a smooth global inverse.

\subsection{The volume identity}

Apply the change of variables $x=\nabla F_t(u)$, whose Jacobian
determinant is $\det(\nabla^2F_t(u))>0$. The boundary of a convex
body has Lebesgue measure zero, so the volume of its interior is its
volume. Consequently,
\begin{equation}\label{eq:volume-integral}
 \Vol_n\!\left(\sum_i t_iQ_i\right)
 =\int_{\R^n}\det\!\left(\sum_i t_iH_i(u)\right)\,du.
\end{equation}
The nonnegative integrand is integrable on the unbounded domain,
since $Q_t$ has finite volume.

\subsection{Nonnegative coefficients}

Fix $u$ and expand the homogeneous polynomial
\[
 \det\!\left(\sum_i t_iH_i(u)\right)
 =\sum_{|\beta|=n}d_\beta(u)t^\beta.
\]
To prove nonnegativity of $d_\beta(u)$, write
$H_i(u)=B_iB_i^{\mathsf T}$ and concatenate these $n\times n$
matrices as $B=[B_1\ \cdots\ B_k]$. Cauchy--Binet applied to
$[\sqrt{t_1}B_1\ \cdots\ \sqrt{t_k}B_k]$ yields
\[
 \det\!\left(\sum_i t_iH_i(u)\right)
 =\sum_{\substack{J\subseteq\{1,\ldots,kn\}\\|J|=n}}
       \det(B_J)^2\prod_{i=1}^k t_i^{m_i(J)}.
\]
Here $B_J$ is the submatrix with columns indexed by $J$, and
$m_i(J)$ is the number of selected columns in $B_i$. Each
coefficient is thus a sum of squares. The equality for positive $t$
extends to a polynomial identity.

Evaluating at $t=(1,\ldots,1)$ and using nonnegativity, we have
\[
 0\leq d_\beta(u)
 \leq\sum_{|\gamma|=n}d_\gamma(u)
 =\det\!\left(\sum_i H_i(u)\right).
\]
By~\eqref{eq:volume-integral}, the last expression has integral
$\Vol_n(Q)<\infty$, with $Q=\sum_iQ_i$. Every coefficient function
is therefore integrable. Integrating the finite sum term by term gives
\[
 p_\tau(t)=\sum_{|\beta|=n}t^\beta
                  \int_{\R^n}d_\beta(u)\,du.
\]
Since the polynomials agree on the positive orthant, their
coefficients agree. In particular,
\begin{equation}\label{eq:coefficient-integral}
 b_\alpha^\tau
 =\int_{\R^n}[t^\alpha]
       \det\!\left(\sum_i t_iH_i(u)\right)\,du.
\end{equation}
This proves the nonnegative integral representation without
differentiating under an improper integral.

\section{A bounded estimator}
\label{sec:estimator}

\begin{proposition}[Bounded weight and ideal sampling]
\label{prop:estimator}
Let $X\sim\Unif(Q)$ and $u=(\nabla F)^{-1}(X)$. The weight
$W_\alpha$ in~\eqref{eq:weight} takes values in $[0,1]$ and has
expectation $b_\alpha^\tau/\Vol_n(Q)\ge1/(eA)$ for the chosen
regularization. This expectation can be estimated to relative error
$\xi$, with failure probability $\delta_s$, from
$O(A\xi^{-2}\log\delta_s^{-1})$ independent ideal samples.
\end{proposition}

We prove the proposition for the continuous reference model.
Finite-bit evaluation and the combined error bounds are established
in Sections~\ref{sec:conditioning}--\ref{sec:output}.

To use the integral in Section~\ref{sec:representation}, we express it
as an expectation with respect to a law that can be sampled. The
resulting random variable lies in $[0,1]$, and a lower bound on its
mean will give relative approximation from independent samples.
For now we assume exact uniform sampling, exact function evaluation,
and an exact volume factor. The boundary and numerical errors are
handled in Sections~\ref{sec:conditioning} and~\ref{sec:output}.

\subsection{The sampling map}

Let the sampling body and convex potential be
\[
  Q=\sum_iQ_i,\qquad F(u)=\sum_i f_i(u).
\]
This sets $t_1=\cdots=t_k=1$ in Section~\ref{sec:representation}.
For $x\in\interior Q$, define $u(x)=(\nabla F)^{-1}(x)$.
The Jacobian determinant of $u\mapsto\nabla F(u)$ is
\[
  \mathcal J(u)=\det\nabla^2F(u)
  =\det\!\left(\sum_iH_i(u)\right)>0,
  \qquad dx=\mathcal J(u)\,du.
\]

For coefficient extraction, introduce formal variables
$z=(z_1,\ldots,z_k)$ independently of the fixed sampling weights,
and put
\[
  P_u(z)=\det\!\left(\sum_i z_iH_i(u)\right)
        =\sum_{|\beta|=n}d_\beta(u)z^\beta,
  \qquad d_\alpha(u)=[z^\alpha]P_u(z).
\]
The point $u$ is fixed when coefficients are extracted, and
$P_u(1,\ldots,1)=\mathcal J(u)$.

\subsection{The normalized coefficient}

Define
\begin{equation}\label{eq:weight}
  W_\alpha(x)=
  \frac{[z^\alpha]\det\!\left(\sum_i z_iH_i(u)\right)}
       {\det\!\left(\sum_i H_i(u)\right)}
  =\frac{d_\alpha(u)}{\mathcal J(u)},\qquad u=u(x).
\end{equation}
The denominator is the Jacobian of $x=\nabla F(u)$, for which all
sampling weights are one. It is therefore the value of the determinant
polynomial at $z=(1,\ldots,1)$ and contains no formal variables.

Nonnegativity of $d_\beta(u)$, proved in
Section~\ref{sec:representation}, implies
\[
  \mathcal J(u)=P_u(1,\ldots,1)=\sum_{|\beta|=n}d_\beta(u),
  \qquad 0\le d_\alpha(u)\le\mathcal J(u),
\]
so $0\le W_\alpha(x)\le1$. At a fixed $x$ it is the proportion
of the determinant belonging to the target coefficient. Defining
these weights for each $\beta$ with $|\beta|=n$ gives
$\sum_{|\beta|=n}W_\beta(x)=1$. Since the boundary of $Q$ has
volume zero, the choice of $W_\alpha$ on the boundary does not
change its uniform expectation.

\subsection{Expectation of the weight}

For $X$ uniform on $Q$, the change of variables $x=\nabla F(u)$
gives
\begin{align*}
  \E W_\alpha(X)
  &=\frac1{\Vol_n(Q)}\int_Q W_\alpha(x)\,dx\\
  &=\frac1{\Vol_n(Q)}\int_{\R^n}
       \frac{d_\alpha(u)}{\mathcal J(u)}\mathcal J(u)\,du\\
  &=\frac1{\Vol_n(Q)}\int_{\R^n}d_\alpha(u)\,du
   =\frac{b_\alpha^\tau}{\Vol_n(Q)}.
\end{align*}
The final equality follows from~\eqref{eq:coefficient-integral}.
The normalization in~\eqref{eq:weight} has cancelled the Jacobian.

We can also describe this calculation through the probability density
of $(\nabla F)^{-1}(X)$, which is
\[
  u\longmapsto\frac{\mathcal J(u)}{\Vol_n(Q)}
\]
on $\R^n$ by~\eqref{eq:volume-integral}. The weight is the quotient
of the coefficient density $d_\alpha(u)$ by the volume density
$\mathcal J(u)$.

\subsection{A lower bound for the mean}

The lower bound for the expectation follows from capacity scaling
and the enlargement estimate. As $B_n\subseteq S$, convexity gives
\[
  Q=S+k\tau B_n\subseteq(1+k\tau)S.
\]
By~\eqref{eq:volume-fraction}, the volumes satisfy
\[
  \Vol_n(Q)\le(1+k\tau)^n\Vol_n(S)
             \le(1+k\tau)^n A b_\alpha.
\]
Monotonicity and full dimensionality give
$b_\alpha^\tau\ge b_\alpha>0$. With
$m_\alpha=\E W_\alpha(X)$, it follows that
\[
  m_\alpha\ge\frac1{A(1+k\tau)^n}
  \ge\frac{e^{-nk\tau}}A.
\]
We have $nk\tau=\varepsilon/(CA)\le1$ for the chosen sufficiently
large constant $C$. Therefore
$m_\alpha\ge1/(eA)$, and we have proved
\begin{equation}\label{eq:mean}
  \boxed{\E_{X\sim\Unif(Q)}W_\alpha(X)
  =\frac{b_\alpha^\tau}{\Vol_n(Q)}\ge\frac1{O(A)}.}
\end{equation}
Since $A=2\Gamma_\alpha$, the reciprocal mean is
$O(\Gamma_\alpha)$ for all $k$. It is polynomial in $n$ for fixed
$k$.

\subsection{The number of samples}

Take independent uniform samples $X_1,\ldots,X_N$ from $Q$.
Set $Y_j=W_\alpha(X_j)$ and
\[
  \overline Y_N=\frac1N\sum_{j=1}^NY_j,
  \qquad\E Y_j=m_\alpha.
\]
The bounds $0\le Y_j\le1$ imply $Y_j^2\le Y_j$, whence
\[
  \operatorname{Var}(Y_j)\le m_\alpha(1-m_\alpha)\le m_\alpha,
  \qquad
  \frac{\operatorname{Var}(\overline Y_N)}{m_\alpha^2}
  \le\frac1{Nm_\alpha}.
\]
The relative variance thus depends on $1/m_\alpha$. To bound the
failure probability, apply the multiplicative Chernoff inequality for
independent variables in $[0,1]$. For $0<\xi<1$, it gives
\[
  \Pr\!\left[|\overline Y_N-m_\alpha|\ge\xi m_\alpha\right]
  \le2\exp\!\left(-\frac{Nm_\alpha\xi^2}{3}\right).
\]
This version does not require indicator random variables. Indeed,
for $0\le y\le1$, convexity gives
$e^{\lambda y}\le1+y(e^\lambda-1)$, and hence
$\E e^{\lambda Y_j}\le\exp(m_\alpha(e^\lambda-1))$.
The usual exponential-moment argument then applies.

Consequently, for a sampling failure allowance $0<\delta_s<1$,
it suffices to choose
\[
  N\ge\left\lceil
       3eA\,\xi^{-2}\log\frac2{\delta_s}
      \right\rceil.
\]
The corresponding ideal estimator, using the exact volume, is
\[
  \widetilde b_\alpha^\tau=\Vol_n(Q)\,\overline Y_N,
  \qquad\E\widetilde b_\alpha^\tau=b_\alpha^\tau,
\]
with relative error at most $\xi$ and failure probability at most
$\delta_s$. Taking $\xi$ and $\delta_s$ to be fixed fractions of
$\varepsilon$ and $\delta$, respectively, gives
$N=O(A\varepsilon^{-2}\log\delta^{-1})$.
Sections~\ref{sec:conditioning}--\ref{sec:output} supply the
approximate evaluation, the boundary truncation, and the estimate of
the volume factor.

\section{Boundary estimates and numerical conditioning}
\label{sec:conditioning}

\subsection{An interior sampling region}

We restrict sampling to the smaller body defined by
\[
  \theta=\frac{\varepsilon}{CnA},\qquad D=(1-\theta)Q.
\]
If $Z$ is uniform on $Q$ and $X=(1-\theta)Z$, then $X$ is uniform
on $D$. Every point of $D$ has distance at least $\theta$ from $\partial Q$.
The bound $W_\alpha\le1$ therefore gives
\begin{equation}\label{eq:boundary}
  \int_{Q\setminus D}W_\alpha(x)\,dx
  \le n\theta\Vol_n(Q)=O(\varepsilon)b_\alpha.
\end{equation}
We next bound the inverse image of this region under the gradient map.
The identity $\sum_i h_i=h_S$ and the inclusion $S\subseteq R_1B_n$
imply $h_S(G)\le R_1\norm G$. It follows that
\[
  F(0)=\sum_i\E h_i(G)+k\tau
      =\E h_S(G)+k\tau
      \le R_1\sqrt n+k\tau.
\]
Thus the Gaussian term has no additional factor $k$.
For $x\in D$, let $u_*=(\nabla F)^{-1}(x)$. We have
\[
  \theta\norm{u_*}\le F(u_*)-x\cdot u_*
  \le F(0)\le R_1\sqrt n+k\tau.
\]
Hence
\begin{equation}\label{eq:radius}
  \norm{u_*}\le U:=\frac{R_1\sqrt n+k\tau}{\theta}
  =\poly(n,k,A,\varepsilon^{-1}).
\end{equation}
The rational samples will be evaluated in the larger region
$E=(1-\theta/2)Q$. Since this region is at distance at least $\theta/2$
from $\partial Q$, its inverse image lies in the ball of radius $2U$ by
the same argument. Taking $R_1$ to be the rational outer bound supplied
by rounding, define
\[
 \overline U=\frac{R_1n+k\tau}{\theta}\ge U,\qquad
 B=4\overline U,\qquad
 \mu=\frac{\tau}{(1+B^2)^2}.
\]
All three constants are rational. We have $\norm{u(x)}\le B/2$ for every $x\in E$,
and~\eqref{eq:curvature} implies $H_i(u)\succeq\mu I_n$ on $BB_n$.
Replacing the exponent by $2$ gives a rational lower bound for the expression
with exponent $3/2$. The quantities $B$ and $\mu^{-1}$ are bounded by polynomials
of fixed degree in $n,k,A,\varepsilon^{-1}$.
We will also use $D+(\theta/4)B_n\subseteq E$. This follows from
$B_n\subseteq Q$ and the fact that positive scalar multiples of the convex
body $Q$ add by adding their scalars. In particular, rounding a sample
within this distance keeps it in the evaluation region.

\subsection{Rational approximation of the Hessians}
\label{sec:numerical-evaluation}

\begin{lemma}[Numerical evaluation]
\label{lem:numerical-evaluation}
Let $x\in E=(1-\theta/2)Q$ be rational, and let
$0<\eta\le1/2$ and $0<\zeta_{\mathrm{eval}}<1/2$ be rational.
There is a randomized subroutine returning rational positive-definite matrices
$\widehat H_i$ for which~\eqref{eq:matrix-accuracy} holds with probability
at least $1-\zeta_{\mathrm{eval}}$, where
$u_*=(\nabla F)^{-1}(x)$.
The number of oracle calls and bit operations, and the query and output
bit lengths, are bounded by a polynomial of absolute degree in $n,k,A$,
$1+\log R_0$, $\varepsilon^{-1}$, $\eta^{-1}$,
$1+\log\zeta_{\mathrm{eval}}^{-1}$, and the bit length of $x$.
These cost bounds hold on every execution, whether or not the accuracy
claim holds.
\end{lemma}

\begin{proof}
We first approximate derivatives at a rational point, then use these
approximations to invert the gradient map. Throughout the proof, support
values are computed to prescribed absolute accuracy by the deterministic
weak-optimization reductions of Section~\ref{sec:model}. The inner balls
needed for these reductions were obtained in Section~\ref{sec:scaling}.

\paragraph{Uniform derivative bounds.}
For $g_i(u)=\E h_i(u+G)$, Gaussian differentiation and centering give
\[
 \nabla g_i(u)=\E\!\left[G\bigl(h_i(u+G)-h_i(u)\bigr)\right]
\]
and
\begin{equation}\label{eq:hessian-estimator}
 \nabla^2g_i(u)=
 \E\!\left[(GG^\top-I_n)\bigl(h_i(u+G)-h_i(u)\bigr)\right].
\end{equation}
The centering terms vanish since $\E G=0$ and $\E(GG^\top-I_n)=0$.
Let $V_u(G)=G(h_i(u+G)-h_i(u))$ and
$A_u(G)=(GG^\top-I_n)(h_i(u+G)-h_i(u))$. The $R_1$-Lipschitz property of $h_i$ gives
the following bounds uniformly in $u$:
\[
 \E\norm{V_u(G)}^2\le R_1^2n(n+2),\qquad
 \E\norm{A_u(G)}_F^2\le R_1^2n(n+1)(n+4).
\]
Here $\norm{\cdot}_F$ is the Frobenius norm. The second estimate follows by
substituting $\norm{GG^\top-I_n}_F^2=\norm G^4-2\norm G^2+n$ and using the
second, fourth, and sixth moments of $\norm G$.
Applying the uncentered Gaussian Hessian identity at two points gives
\begin{equation}\label{eq:hessian-lipschitz}
 \norm{H_i(u)-H_i(v)}_{\op}\le L_H\norm{u-v},
 \qquad L_H=R_1(n+1)+6\tau.
\end{equation}
For this estimate, use $\E\norm{GG^\top-I_n}_{\op}\le n+1$ and the fact
that the derivative of the Hessian of $\sqrt{1+\norm u^2}$ has operator
norm at most $6$. The moment estimate also gives, for all points,
$\norm{H_i(u)}_{\op}\le B_H=R_1(n+4)^2+\tau$.

\paragraph{A finite-bit derivative primitive.}
Fix rational $u$ and $0<s\le1$. We construct an approximation to
$\nabla f_i(u)$ or $H_i(u)$ whose Euclidean or Frobenius error, respectively,
is at most $s$, except with probability $\rho$. The cost will be polynomial
in the query length, $n$, $R_1$, $s^{-1}$, and $\log\rho^{-1}$.

Let $J_0=4R_1(n+4)^2$ and choose an integer $M\ge1$ such that
\[
 2n e^{-M^2/2}\le \left(\frac{s}{16J_0}\right)^2.
\]
For example, $M=2\lceil\log_2(64n(J_0+1)/s)\rceil$ suffices.
Replace each coordinate of $G$ outside $[-M,M]$ by the nearer endpoint,
and denote the resulting vector by $G^{\mathrm c}$.
This changes the vector with probability at most $2ne^{-M^2/2}$.
By Cauchy--Schwarz, the resulting bias in either score is at most
$J_0$ times the square root of this probability, and hence at most $s/16$.
The required bound for the clipped score follows from
$\norm{G^{\mathrm c}}\le\norm G$ and
\[
 \norm{V_u(g)}\le R_1\norm g^2,\qquad
 \norm{A_u(g)}_F\le R_1(\norm g^2+\sqrt n)\norm g.
\]
The same Gaussian moments bound the sum of the two $L^2$ norms,
before and after clipping, by $J_0$.

Put
\[
 T_G=n(M+1)+1,\qquad
 L_G=R_1(3T_G^2+n+2T_G+1),\qquad
 \omega=\min\{1,s/(16L_G)\}.
\]
On the radius-$T_G$ ball, the gradient score is $2R_1T_G$-Lipschitz
and the Hessian score is $R_1(3T_G^2+n)$-Lipschitz in Frobenius norm.
It follows that replacing the clipped Gaussian by a rational $\widehat G$
with a coupling satisfying $\norm{\widehat G-G^{\mathrm c}}\le\omega$
changes either expectation by at most $s/16$.

To generate this vector in bounded time, apply the clipped standard-normal
quantile function to independent $b$-bit midpoint uniforms, taking
\[
 b\ge M^2+2+\left\lceil\log_2(4n/\omega)\right\rceil,
\]
and compute each quantile to absolute error at most $\omega/(2n)$.
The clipped quantile map is $2^{M^2+2}$-Lipschitz, since the normal
density on $[-M,M]$ is bounded below by $2^{-M^2-2}$.
Coupling each midpoint uniform to its continuous uniform gives the
claimed spatial approximation for $\widehat G$.
To compute the normal CDF on $[-M,M]$, integrate the Taylor expansion
of $e^{-t^2/2}$. Accuracy $a$ uses $O(M^2+\log a^{-1})$ terms and
arithmetic of polynomial bit length. Bisection using CDF error at most
$\omega 2^{-M^2-2}/(32n)$ then gives the required quantile accuracy.
When a comparison cannot be resolved, the lower bound for the density
implies that the spatial error is already small enough to terminate.
The constants $\pi$ and the square roots are computed to the same
polynomial precision. We have thus used a bounded number of fair bits
and a spatial coupling, without requiring a total-variation comparison
between discrete and continuous Gaussian laws.

For every score, approximate each of $h_i(u+\widehat G)$ and $h_i(u)$
with absolute error at most
\[
 e_h=\frac{s}{32(T_G^2+T_G+n+1)}.
\]
The error in either score is then at most $s/16$ on every sample.
It remains to bound the deviation of the empirical scores with exact
support values. Every coordinate of such a score at $\widehat G$ is
bounded in absolute value by $J=R_1(T_G+1)(T_G^2+n+1)$. Take
\[
 q\ge \frac{32n^2J^2}{s^2}\log\frac{2n^2}{\rho}
\]
independent samples. Hoeffding's inequality, followed by a union bound
over at most $n^2$ coordinates, gives coordinate error at most $s/(4n)$.
Thus the Euclidean or Frobenius error is at most $s/4$, except with
probability $\rho$. We implement this choice using integer upper bounds
for $q$ and for the logarithmic quantities.

We add the gradient or Hessian of $\tau\sqrt{1+\norm u^2}$, computed
with norm error at most $s/16$, and round the entries to dyadic rationals
with norm error at most $s/16$. Symmetric entries receive the same value.
Summing the errors from clipping, quantization, support evaluation,
sampling, deterministic evaluation, and rounding gives
$s/16+s/16+s/16+s/4+s/16+s/16<s$.
The construction is uniform in $u$. Its random-string lengths, numbers
of iterations, query lengths, and output lengths have polynomial bounds
on all executions, including failed ones.

\paragraph{Approximation of the inverse point.}
Let $B,\mu$ be the rational constants of Section~\ref{sec:conditioning},
and set $m=k\mu$ and $L=kB_H$.
The function $\varphi_x(u)=F(u)-x\cdot u$ on the radius-$B$ ball satisfies
\[
 mI_n\preceq\nabla^2\varphi_x(u)\preceq LI_n,
 \qquad \norm{u_*}\le B/2.
\]
Set
\[
 d=\frac{\eta\mu}{8L_H},\qquad
 e=\frac{md}{4},\qquad p=\frac{md}{4L},\qquad
 J_{\rm iter}=\left\lceil\frac Lm\log\frac{2B}{d}\right\rceil.
\]
We may use an integer upper bound for this iteration count.
Start at $u_0=0$. At each iteration approximate the gradients
$\nabla f_i(u_j)$ to tolerance $e/k$, sum them, and subtract the rational
vector $x$. After a step of length $1/L$, approximate the projection
onto $BB_n$ by a rational point in the ball with error at most $p$.
For a rational tentative point $v$, compute
$\min\{1,B/\norm v\}$ from below and multiply by $v$; when $v=0$,
no approximation is needed. Rational norm comparisons and square-root
approximation achieve this accuracy with polynomially many bits.

Condition on the derivative approximations being accurate. The exact
gradient step is a contraction, and projection is nonexpansive. Hence
\[
 \norm{u_{j+1}-u_*}
 \le(1-m/L)\norm{u_j-u_*}+e/L+p.
\]
After the prescribed number of iterations, the initial error contributes
at most $d/4$, while the perturbations contribute at most
$e/m+pL/m=d/2$. We therefore obtain $\norm{u_{J_{\rm iter}}-u_*}\le d$.
The derivative primitive uses dyadic rounding. We also take the rational
projection scales to be dyadic. Their bit lengths are bounded by the
tolerance and the magnitude bounds, independently of the iteration.
It follows that each step increases the denominator length by at most
a polynomial amount. Since there are polynomially many steps, all
query lengths remain polynomial; a common fixed grid for the iterates
is unnecessary.

\paragraph{Matrix accuracy and failure probability.}
Apply the derivative primitive at the last iterate to approximate each
$H_i$ in Frobenius norm with tolerance $s=\eta\mu/4$.
On the event that these approximations are accurate,
\[
 \norm{\widehat H_i-H_i(u_*)}_{\op}
 \le \eta\mu/4+L_Hd=3\eta\mu/8<\eta\mu.
\]
Since $H_i(u_*)\succeq\mu I_n$, this implies
\begin{equation}\label{eq:matrix-accuracy}
 (1-\eta)H_i(u_*)\preceq\widehat H_i\preceq(1+\eta)H_i(u_*).
\end{equation}
For each call, set its failure probability to
$\rho=\zeta_{\mathrm{eval}}/[k(J_{\rm iter}+1)]$.
There are $kJ_{\rm iter}$ gradient calls and $k$ final Hessian calls,
each using fresh random bits. The current iterate always lies in $BB_n$,
and the estimates for the primitive are uniform. Its guarantee therefore
holds conditional on any preceding history. The union bound gives
failure probability at most $\zeta_{\mathrm{eval}}$.

Test the resulting matrices for positive definiteness using exact rational
$LDL^\top$ elimination, and replace any matrix that fails this test by
$\mu I_n$. On the successful event there are no replacements. On all
other executions, the outputs are still positive definite and have
polynomial bit length. In particular, their normalized determinant
coefficient is defined and lies in $[0,1]$.
The construction has polynomial bounds in the parameters of the lemma.
Dependence on the inverse failure probability occurs only logarithmically,
through the sample counts. This proves the two assertions.
\end{proof}

\paragraph{Perturbation of the normalized weight.}
We next apply mixed-discriminant monotonicity. Under the replacement of
$H_i$ by $\widehat H_i$, the ratio~\eqref{eq:weight} changes in either
direction by a factor of at most
\begin{equation}\label{eq:ratio-accuracy}
 \left(\frac{1+\eta}{1-\eta}\right)^n=1+O(\varepsilon)
 \quad\text{when }\eta=O(\varepsilon/n).
\end{equation}
To see monotonicity, increase one argument by a positive-semidefinite
matrix. By multilinearity, the change is the mixed discriminant with
this increment as the corresponding argument. It is nonnegative by
Section~\ref{sec:representation}.
Applying this fact to the $n$ repeated arguments bounds the numerator
between $(1-\eta)^n$ and $(1+\eta)^n$ times its original value.
The same bounds hold for the denominator, giving~\eqref{eq:ratio-accuracy}.

\begin{lemma}[Lipschitz continuity of the normalized weight]
\label{lem:weight-lipschitz}
Let $E=(1-\theta/2)Q$, and let $\mu$ and $L_H$ be the constants from
Section~\ref{sec:conditioning} and~\eqref{eq:hessian-lipschitz}.
For every $\alpha\in\mathbb Z_{\ge0}^k$ with $|\alpha|=n$, the normalized weight
$W_\alpha$ in~\eqref{eq:weight} is smooth and strictly positive on
$\interior Q$. On $E$ it satisfies
\begin{equation}\label{eq:weight-lipschitz}
 \bigl|W_\alpha(x)-W_\alpha(y)\bigr|
 \le L_W\norm{x-y},\qquad
 L_W:=\frac{2nL_H}{k\mu^2}=\poly(n,k,A,\varepsilon^{-1}).
\end{equation}
In fact, the same bound holds for
$|\log W_\alpha(x)-\log W_\alpha(y)|$.
Consequently, for any two probability measures $\nu,\nu'$ supported in $E$,
\begin{equation}\label{eq:weight-wasserstein}
 \left|\int_E W_\alpha\,d\nu-\int_E W_\alpha\,d\nu'\right|
 \le L_W\mathsf W_1(\nu,\nu'),
\end{equation}
where $\mathsf W_1$ denotes the Wasserstein distance with Euclidean cost:
\[
 \mathsf W_1(\nu,\nu')
 =\inf_{\pi\in\Pi(\nu,\nu')}
   \int_{E\times E}\norm{x-y}\,d\pi(x,y),
\]
and $\Pi(\nu,\nu')$ is the set of couplings of the two measures.
In particular,~\eqref{eq:weight-wasserstein} applies with
$\nu'=\Unif(D)$, since $D=(1-\theta)Q\subseteq E$.
\end{lemma}

\begin{proof}
The inverse map $u=(\nabla F)^{-1}$ is smooth on $\interior Q$.
Positive definiteness of the $H_i(u)$ implies $d_\alpha(u)>0$ by
mixed-discriminant monotonicity. Indeed, for fixed $u$, take $a>0$ such
that $H_i(u)\succeq aI_n$ for every $i$. Comparison with
$\det(\sum_i z_i aI_n)=a^n(\sum_i z_i)^n$ proves the assertion.
Since $\mathcal J(u)$ is positive as well, $W_\alpha$ is smooth and
strictly positive. Section~\ref{sec:estimator} gives $W_\alpha\le1$.

Let $x,y\in E$, $v=y-x$, $x_s=x+sv$, and $u_s=u(x_s)$,
where $0\le s\le1$. By convexity of $E$ and the estimates of
Section~\ref{sec:conditioning}, $\norm{u_s}\le B/2$ and
$H_i(u_s)\succeq\mu I_n$. Differentiation of $\nabla F(u_s)=x_s$ gives
\[
 \dot u_s=\left(\sum_i H_i(u_s)\right)^{-1}v,
 \qquad \norm{\dot u_s}\le\frac{\norm v}{k\mu}.
\]
Write $M_i(s)=H_i(u_s)$ and $c=L_H\norm v/(k\mu^2)$.
The Hessian Lipschitz bound~\eqref{eq:hessian-lipschitz} gives
\[
 \norm{\dot M_i(s)}_{\op}\le\frac{L_H\norm v}{k\mu}=c\mu,
 \qquad -cM_i(s)\preceq\dot M_i(s)\preceq cM_i(s).
\]
Set $N(s)=[z^\alpha]\det(\sum_i z_iM_i(s))$ and
$J(s)=\det(\sum_iM_i(s))$.
The function $N$ is a positive constant times a mixed discriminant with
$n$ arguments, in which $M_i$ occurs $\alpha_i$ times.
Differentiating gives $n$ terms, one for each differentiated argument.
Multilinearity and positivity, together with the preceding spectral
bounds, show that each term has absolute value at most $cN$.
For $J$ the same reasoning applies because
$-c\sum_iM_i\preceq\sum_i\dot M_i\preceq c\sum_iM_i$.
It follows that
\begin{align*}
 \left|\frac{d}{ds}\log N(s)\right|&\le nc,\qquad
 \left|\frac{d}{ds}\log J(s)\right|\le nc,\\
 \left|\frac{d}{ds}\log W_\alpha(x_s)\right|
 &\le2nc=L_W\norm v.
\end{align*}
Integration gives the logarithmic estimate. Using $0<W_\alpha\le1$ in
$\frac{d}{ds}W_\alpha=W_\alpha\frac{d}{ds}\log W_\alpha$ also proves
\eqref{eq:weight-lipschitz}.
Now let $\pi$ be any coupling of $\nu$ and $\nu'$.
Integrating the Lipschitz inequality against this coupling yields
\[
 \left|\int_E W_\alpha\,d\nu-\int_E W_\alpha\,d\nu'\right|
 \le L_W\int_{E\times E}\norm{x-y}\,d\pi(x,y).
\]
Taking the infimum over couplings proves~\eqref{eq:weight-wasserstein}.
\end{proof}

\subsection{Exact coefficient computation}
\label{sec:determinant-arithmetic}
\label{sec:exact-coefficient}

We compute the determinant coefficient by mixed finite differences,
with a complexity bound uniform in $k$. The number of evaluations is $\prod_i(\alpha_i+1)$, which we bound
polynomially in the multiplicity-sensitive Gurvits ratio
\[
 \Gamma_\alpha=\frac{n^n}{n!}
                  \prod_{i=1}^k\frac{\alpha_i!}{\alpha_i^{\alpha_i}}.
\]
As before, we discard bodies of zero multiplicity. Thus $1\le k\le n$
and $\alpha_i\ge1$.

\begin{lemma}[Size of the coefficient grid]
\label{lem:coefficient-grid}
For positive integers $\alpha_1,\ldots,\alpha_k$ with sum $n$, put
$\mathcal N_\alpha=\prod_i(\alpha_i+1)$.  Then
\begin{equation}\label{eq:coefficient-grid-bound}
 \mathcal N_\alpha\le
  2\pi n\exp\!\left(\frac1{6n}\right)\Gamma_\alpha^2
  <8n\Gamma_\alpha^2,
 \qquad \mathcal N_\alpha\le2^n.
\end{equation}
The positive rational number $\Gamma_\alpha$ can be computed exactly with
numerator and denominator of $O(n\log(n+1))$ bits.
\end{lemma}

\begin{proof}
Write $s_m=m!e^m/m^m$.  The usual two-sided Stirling bounds give
\[
 \sqrt{2\pi m}\le s_m
 \le\sqrt{2\pi m}\exp\!\left(\frac1{12m}\right),\qquad m\ge1.
\]
Since $m+1\le2\pi m\le s_m^2$ and
$\Gamma_\alpha=(\prod_i s_{\alpha_i})/s_n$, we obtain
\[
 \mathcal N_\alpha\le\prod_i s_{\alpha_i}^2
 =\Gamma_\alpha^2s_n^2
 \le2\pi n\exp\!\left(\frac1{6n}\right)\Gamma_\alpha^2.
\]
The factor multiplying the squared ratio is less than $8n$ when $n\ge1$.
The inequality $m+1\le2^m$ proves the second bound.
To bound the encoding length, write
\[
 \Gamma_\alpha=
 \frac{n^n\prod_i\alpha_i!}{n!\prod_i\alpha_i^{\alpha_i}}.
\]
Both products have $O(n\log(n+1))$ bits, since
$\sum_i\alpha_i=n$ and $\alpha_i\le n$. Reduction of the fraction
does not increase its encoding length.
\end{proof}

\begin{proposition}[Exact computation from rational Hessians]
\label{prop:exact}
Let $k$ vary, let $\alpha_i\ge1$ and $|\alpha|=n$, and let
$\widehat H_1,\ldots,\widehat H_k$ be rational symmetric positive-definite
$n\times n$ matrices.  Suppose that each entry numerator and positive
denominator has at most $b_{\mathrm{in}}$ bits.  The coefficient
\[
 \widehat N=[z^\alpha]\det\!\left(\sum_i z_i\widehat H_i\right)
\]
and its normalized weight
$\widehat W_\alpha=\widehat N/\det(\sum_i\widehat H_i)$ can be computed
exactly using at most $\mathcal N_\alpha$ determinant evaluations and
\[
 O\!\left(\mathcal N_\alpha\,
       \poly(n,k,b_{\mathrm{in}})\right)
 \subseteq \poly(n,k,b_{\mathrm{in}},\Gamma_\alpha)
\]
bit operations.  The polynomial degrees are absolute and do not depend on
$k$.  The coefficient and normalized weight have polynomial bit lengths
in $n,k,b_{\mathrm{in}}$.
\end{proposition}

\subsubsection{Mixed finite differences}
\label{sec:finite-differences}

Let $q_0$ be a common positive denominator of the matrix entries, and set
$M_i=q_0\widehat H_i\in\mathbb Z^{n\times n}$. We may take the product
of all entry denominators for this purpose. Put
\[
 P(z)=\det\!\left(\sum_i z_iM_i\right),\qquad
 \nu_\alpha=[z^\alpha]P(z),\qquad
 \alpha!=\prod_i\alpha_i!.
\]
For $0\le j_i\le\alpha_i$ write
$|j|=\sum_i j_i$ and
$\binom\alpha j=\prod_i\binom{\alpha_i}{j_i}$.  Then
\begin{equation}\label{eq:mixed-finite-difference}
 \boxed{\displaystyle
 \nu_\alpha=
 \frac1{\alpha!}\sum_{0\le j\le\alpha}
       (-1)^{n-|j|}\binom\alpha j
       \det\!\left(\sum_i j_iM_i\right).}
\end{equation}
The denominator is the product of the multiplicity factorials.
For $\alpha=(1,\ldots,1)$ it is $1$, whereas for $M_i=M$ the resulting
coefficient is $\nu_\alpha=\binom n\alpha\det M$, with
$\binom n\alpha=n!/\alpha!$.

To prove the identity, recall that the $a$th forward difference of $t^b$
at zero satisfies
\[
 \sum_{j=0}^{a}(-1)^{a-j}\binom aj j^b
 =0\quad(b<a),\qquad =a!\quad(b=a).
\]
Apply the product of the difference operators to a monomial $z^\beta$
of the homogeneous degree-$n$ polynomial $P$. Its contribution vanishes
unless $\beta_i\ge\alpha_i$ for every $i$. Since
$|\beta|=|\alpha|=n$, the only surviving monomial has $\beta=\alpha$.
Its contribution is $\alpha!\nu_\alpha$, as required
in~\eqref{eq:mixed-finite-difference}.
With the usual convention for a zeroth difference, the proof also permits
zero multiplicities; those bodies have already been removed in our algorithm.

The formula~\eqref{eq:mixed-finite-difference} involves integer matrices.
If $j\ne0$, then $\sum_i j_iM_i$ is positive definite and all its leading
principal minors are positive. We can therefore use fraction-free Bareiss
elimination~\cite{Bareiss}, with nonzero pivots and exact integer divisions.
For $j=0$, we set the determinant of the zero matrix to zero without
elimination. The same rational Hessian approximations are used at all
nodes, so the cancellations in the signed sum introduce no error.

The identity~\eqref{eq:mixed-finite-difference} shows that its signed sum
is divisible by $\alpha!$. After this exact division, define
\[
 \Delta=P(1,\ldots,1),\qquad
 \widehat N=\frac{\nu_\alpha}{q_0^n},\qquad
 \widehat W_\alpha=\frac{\nu_\alpha}{\Delta}.
\]
Since every $\alpha_i\ge1$, the all-ones node is in the grid and
$\Delta$ has already been computed. The Cauchy--Binet formula applied
to Gram factorizations gives nonnegative determinant coefficients for
positive-semidefinite arguments. Expanding
$M_i=cI+(M_i-cI)$ for a sufficiently small common $c>0$ shows further
that every degree-$n$ coefficient is positive when all matrices are
positive definite. Consequently $\Delta>0$ and $\nu_\alpha>0$.
Nonnegativity of the remaining coefficients gives $0<\nu_\alpha\le\Delta$.

\subsubsection{Uniform bit complexity}
\label{sec:coefficient-bit-complexity}

The common denominator obeys
\[
 \log_2q_0\le kn^2 b_{\mathrm{in}},
\]
and the entries of $M_i$ have
$B=O(kn^2 b_{\mathrm{in}}+1)$ bits.  Since
$\sum_i j_i\le n$ at every grid node, its matrix entries have
$O(B+\log(n+1))$ bits.  Hadamard's inequality bounds every determinant
and every minor used in fraction-free elimination by
\[
 O\!\left(n(B+\log(n+1))\right)
\]
bits.  Products formed during elimination have bit lengths of the same
polynomial order. It follows that each determinant can be computed in
$\poly(n,k,b_{\mathrm{in}})$ bit operations, with an absolute degree.

Every integer binomial weight has at most $n+1$ bits, and
\[
 \sum_{0\le j\le\alpha}\binom\alpha j
 =\prod_i2^{\alpha_i}=2^n.
\]
A partial signed sum therefore has $O(n(B+\log(n+1))+n)$ bits,
since its absolute value is at most $2^n$ times the largest determinant
magnitude. The bound $\alpha!\le n!$ gives an encoding length of
$O(n\log(n+1))$ bits. Summation, exact division by $\alpha!$, and
formation of the output fractions thus have polynomial cost per node,
uniformly in the number of bodies. Proposition~\ref{prop:exact} now
follows from Lemma~\ref{lem:coefficient-grid}.

Only the replacement of the Hessians by rational matrices contributes
approximation error; this was bounded in~\eqref{eq:matrix-accuracy}.
The subsequent coefficient computation and normalization are exact.
If we instead round the coefficient, its positivity and integrality give
$\nu_\alpha\ge1$, and therefore $\widehat N\ge q_0^{-n}$.
It is enough to retain
\[
 \left\lceil n\log_2q_0+\log_2(1/\rho)\right\rceil+1
\]
fractional binary bits to obtain relative rounding error at most $\rho$.

\section{The estimator and its running time}
\label{sec:output}

We first prove a success probability of at least $3/4$ for one complete
run. Independent repetitions and a median will reduce the failure
probability to $\delta$.
Take $C=4096$ in Sections~\ref{sec:balls} and~\ref{sec:conditioning}, and set
\[
 \kappa=\varepsilon/16,\qquad
 \eta=\kappa/(8n),\qquad
 \sigma=\min\{\theta/4,\ k\eta\mu^2/L_H\},\qquad
 N=\left\lceil72A\kappa^{-2}\right\rceil.
\]
All parameters are rational, and $N$ is a positive integer.

\subsection{Rational sampling}
\label{sec:rational-sampling}

For each sample, we require a fresh run producing a rational point
$\widehat X\in E$ and a coupling to $X\sim\Unif(D)$ such that
\[
 \Pr\{\norm{\widehat X-X}>\sigma\}\le\frac1{32N}.
\]
We obtain this guarantee from the finite-grid sampler of
Dyer--Frieze--Kannan~\cite[Theorem~1 and
Sections~2--3]{DFK}.
Their sampler mixes on coarse cubes, selects a fine-grid point uniformly
inside the chosen cube, and applies rejection. Weak membership is
permitted by Section~6, Remark~2 of that paper. We use its mixing
estimate and give the quantization argument needed for the present
spatial coupling. In particular, we do not compare a discrete sample
directly to a continuous law in total variation.

Let $K$ satisfy $rB_n\subseteq K$. Subdivide each coarse cube into an
equal number of half-open dyadic cells of diameter at most $d_0$.
Fix deterministic weak predicates, including tie-breaking, for the coarse
walk and the acceptance test. Use tolerance $\rho_0$ to test each proposed
lower corner. The accepted set is consequently fixed, and under the ideal
stationary coarse law the accepted corners are uniform on this set.
For the coupling argument only, displace a corner uniformly in its cell.
If $J$ is the union of accepted cells, the displaced point is uniform
on $J$. Writing $s_0=\rho_0+d_0$, we have
\[
 K^{-s_0}\subseteq J\subseteq K^{+s_0},\qquad
 (1-s_0/r)K\subseteq K^{-s_0},\qquad
 K^{+s_0}\subseteq(1+s_0/r)K.
\]
For the first inclusion, use a coarse covering containing every cube
that meets $K$, and a fine subdivision smaller than the enlargement
scale of this covering. Every point of $K^{-s_0}$ then lies in a covered
cell whose corner belongs to $K^{-\rho_0}$, so that corner is accepted.
Cell boundaries have measure zero. When $ns_0/r$ is sufficiently small,
the displayed inclusions bound the total-variation distance between
the continuous laws $\Unif(J)$ and $\Unif(K)$ by $O(ns_0/r)$.
The displacement from a corner has length at most $d_0$. Taking
$d_0\le\sigma$ and then choosing $s_0$ sufficiently small gives the
desired coupling. We allocate fixed portions of the failure probability
to this boundary error, to mixing of the coarse walk, and to rejection
exceeding its prescribed cap. After rounding, the acceptance probability
has a constant lower bound; thus a cap logarithmic in the inverse
allowance suffices. The tolerances are transported through the sampler's
rational affine rounding map.

We apply this construction with $K=D$, whose inradius is at least
$1/2$ and whose outer radius is polynomial. In the original coordinates,
take the final weak acceptance tolerance at most $\theta/4$.
Every accepted corner then belongs to $D^{+\theta/4}\subseteq E$.
If rejection reaches its cap, return $0\in E$ and include this event
in the failure allowance. Thus the returned point belongs to $E$ on
every execution.
All grid and oracle parameters have polynomial bit lengths. For one
run, a polynomial bound in $n$, the presentation length, $\sigma^{-1}$,
and the inverse coupling allowance suffices, since these quantities
are polynomially bounded. Oracle simulation costs are counted using
the original input bodies, as in Section~\ref{sec:model}.

We now bound the effect of this spatial error on $W_\alpha$.
The relative matrix estimates in the proof of Lemma~\ref{lem:weight-lipschitz}
give the required multiplicative bound. On the successful coupling event,
strong convexity on the inverse-image ball implies
\[
 \norm{u(\widehat X)-u(X)}\le\frac{\sigma}{k\mu},\qquad
 \norm{H_i(u(\widehat X))-H_i(u(X))}_{\op}
 \le\frac{L_H\sigma}{k\mu}\le\eta\mu.
\]
Using $H_i(u(X))\succeq\mu I_n$, we obtain the relative bounds
of~\eqref{eq:matrix-accuracy} for the spatial perturbation.
Equation~\eqref{eq:ratio-accuracy} therefore applies both to this
perturbation and to the evaluation error of Lemma~\ref{lem:numerical-evaluation}.
Each changes the weight by a factor in $[1-\kappa,1+\kappa]$, since
\[
 \log\left(\frac{1+\eta}{1-\eta}\right)^n
 \le4n\eta=\kappa/2,\qquad e^{\kappa/2}\le1+\kappa,
\]
and the reciprocal is at least $1-\kappa$.

\subsection{Accuracy of one run}

Assign total failure probability $1/32$ to capacity scaling and rounding.
Estimate $\Vol_n(Q)$ with relative error $\kappa$ and failure probability
$1/32$. Generate $N$ independent rational samples by the preceding
construction. At each sample, apply Lemma~\ref{lem:numerical-evaluation}
with $\eta$ and $\zeta_{\mathrm{eval}}=1/(32N)$, and compute the normalized
weight exactly from the resulting rational matrices. The estimator is
\begin{equation}\label{eq:final-estimator}
 \widehat b_\alpha=(1-\theta)^n\widehat{\Vol}_n(Q)
 \frac1N\sum_{j=1}^N\widehat W_\alpha(\widehat X_j).
\end{equation}
Divide the preprocessing allowance among the volume calls used in
capacity optimization and the randomized rounding calls. Give each
subroutine a polynomial step limit and a rational output in case of
failure. If only an expected polynomial bound is available, impose a
cutoff and charge its probability to that subroutine's allowance.
For a constant-success run, the inverse allowances are polynomially
bounded. The deterministic oracle reductions require no allowance.
If preprocessing reports failure or returns malformed data, return zero;
otherwise continue subject to the fixed cost bounds. Also return zero
if the final volume estimate is nonpositive or malformed. An inaccurate
result not detected by these tests belongs to the corresponding
subroutine's failure event. For numerical evaluation,
Lemma~\ref{lem:numerical-evaluation} already supplies a positive-definite
fallback output.

Fix any successful preprocessing realization for the remainder of the
single-run analysis. Write $b=b_\alpha$ and set
\[
 I_D=\int_D W_\alpha(x)\,dx,\qquad r_\varepsilon=2\varepsilon/C.
\]
The choice $nk\tau=\varepsilon/(CA)$ and the bounds for regularization
and the omitted boundary layer imply
\[
 \Vol_n(Q)\le2Ab,\qquad
 0\le b_\alpha^\tau-b\le r_\varepsilon b,\qquad
 0\le b_\alpha^\tau-I_D\le r_\varepsilon b.
\]
These estimates use $e^t-1\le2t$ at $t=nk\tau$, together with
$\Vol_n(Q\setminus D)\le n\theta\Vol_n(Q)$. Hence
\[
 (1-r_\varepsilon)b\le I_D\le(1+r_\varepsilon)b,\qquad
 m_D:=\frac{I_D}{\Vol_n(D)}\ge\frac1{4A}.
\]
The sampler runs are independent. Taking the product of their couplings
therefore gives independent ideal points $X_1,\ldots,X_N$.
For $Y_j=W_\alpha(X_j)\in[0,1]$, the multiplicative Chernoff bound gives
\[
 \Pr\!\left\{\left|N^{-1}\sum_jY_j-m_D\right|>\kappa m_D\right\}
 \le2\exp\!\left(-\frac{N\kappa^2}{12A}\right)<\frac1{32}.
\]
This estimate is conditional on preprocessing alone; no conditioning on
the later coupling or evaluation events is used.

There are five failure allowances: preprocessing, volume estimation,
sample couplings, evaluations, and concentration of the ideal samples.
Each is at most $1/32$. For the evaluations this follows from their
conditional guarantees and a union bound. The probability that all
five conditions hold is consequently at least $27/32>3/4$.
On this event, the two multiplicative perturbation bounds give
\[
 (1-\kappa)^2Y_j\le\widehat W_\alpha(\widehat X_j)
 \le(1+\kappa)^2Y_j.
\]
Use these inequalities, the volume and concentration estimates, and
$\Vol_n(D)=(1-\theta)^n\Vol_n(Q)$ to obtain
\[
 (1-r_\varepsilon)(1-\kappa)^4
 \le\frac{\widehat b_\alpha}{b_\alpha}
 \le(1+r_\varepsilon)(1+\kappa)^4.
\]
Because $r_\varepsilon\le\kappa=\varepsilon/16$ and
$0<\varepsilon<1/2$, the lower bound is at least
$(1-\kappa)^5\ge1-5\kappa\ge1-\varepsilon$ and the upper bound is at most
$(1+\kappa)^5\le e^{5\kappa}\le1+10\kappa\le1+\varepsilon$.
The single-run accuracy claim follows.

\subsection{Amplification and total cost}

Undo the affine transformation and the scaling in each run by returning
\[
 \widehat V=
 \frac{\widehat b_\alpha}
 {|\det T|\lambda^\alpha\binom{n}{\alpha_1,\ldots,\alpha_k}}.
\]
The conversion factors are exact rationals. Repeat the complete procedure
independently an odd number $M_{\rm rep}$ of times, with
\[
 M_{\rm rep}\ge8\log(1/\delta),
\]
and take the median of the resulting $\widehat V$ values.
For example, $M_{\rm rep}=2\lceil4\log_2(2/\delta)\rceil+1$ suffices.
The inverse transformations must be applied before the median, since
$T$ and $\lambda$ can differ between runs. Each run has success
probability at least $3/4$. Hoeffding's inequality therefore bounds
the probability of median failure by $e^{-M_{\rm rep}/8}\le\delta$.
The median is computed exactly over the rationals.

\paragraph{Uniform dependence on the ratio.}
We give explicit bounds to show that the polynomial degree is independent
of $k$. Using $k\le n$, $A=2\Gamma_\alpha\ge2$, and
$0<\varepsilon<1/2$ in the definitions gives
\begin{align*}
 R_1&=O(n^2),&
 \tau^{-1}&=O(n^2A\varepsilon^{-1}),&
 \theta^{-1}&=O(nA\varepsilon^{-1}),\\
 B&=O(n^4A\varepsilon^{-1}),&
 \mu^{-1}&=O(n^{18}A^5\varepsilon^{-5}),&
 L_H&=O(n^3),\\
 B_H&=O(n^4),&
 N&=O(A\varepsilon^{-2}),&
 \sigma^{-1}&=O(n^{40}A^{10}\varepsilon^{-11}).
\end{align*}
The implicit constants are absolute. We have not attempted to optimize
the dimension exponents. The inverse step size, iteration and Gaussian
sample counts, and reciprocal support accuracies used for numerical
evaluation are obtained from these quantities by a fixed number of
additions, multiplications, reciprocals, and logarithms. They too have
polynomial bounds of fixed degree.
Although scaling can produce exponentially small individual inradii,
their logarithms and encoding lengths are polynomial in
$n,k,1+\log R_0$. It is these encoding lengths that enter the weak-oracle
reductions; no polynomial dependence on the reciprocal inradii is required.

A single complete run has cost at most
\[
 T_{\rm scale}+T_{\rm round}+T_{\rm volume}
 +N\bigl(T_{\rm sample}+T_{\rm eval}
       +\mathcal N_\alpha\poly(n,k,b_{\rm in})\bigr).
\]
The bound for $T_{\rm scale}$ in Lemma~\ref{lem:uniform-capacity} requires
no grid exponential in $k$. Lemma~\ref{lem:coefficient-grid} gives
$\mathcal N_\alpha<8n\Gamma_\alpha^2$. In every term we include the cost
of simulating weak-oracle queries through the original membership oracles.
The preceding bounds show that both the presentation lengths and
$b_{\rm in}$ are polynomial. Thus the cost of one run is
\[
 \poly(n,k,1+\log R_0,\Gamma_\alpha,\varepsilon^{-1}),
\]
with an absolute degree. The number of determinant evaluations, obtained
by multiplying the outer sample count by the grid size, is
$O(n\Gamma_\alpha^3\varepsilon^{-2})$ per run. The remaining numerical
work is included separately in the polynomial cost above.

Impose fixed caps on all iteration counts and query and output lengths.
If a query or intermediate encoding exceeds its cap, terminate that
run and return zero. None of these caps is reached on the good event
proved above. They ensure the claimed running time for every random
bit string, including those producing inaccurate preprocessing data.
When the $N$ rational weights are added exactly, the denominator length
is at most the sum of their individual lengths, and is therefore
polynomial. Undoing the transformations also takes polynomial bit time.
For the median, use deterministic linear-time selection on the rational
outputs; each comparison has the cost bound for those outputs.
There are $O(1+\log\delta^{-1})$ independent runs, so the total cost has
the form~\eqref{eq:full-dimensional-complexity}. In particular, there is
no polynomial dependence on $\delta^{-1}$.
This proves Theorem~\ref{thm:full-dimensional}. Applying the exact
reduction of Section~\ref{sec:sections-reduction} proves Theorem~\ref{thm:main}.

For several targets, use $\delta$ divided by the list length for each
one and add their costs, as in Section~\ref{sec:model}.

\paragraph{Volume algorithms used as subroutines.}
For volume and sampling, the polynomial-time guarantees of
Dyer--Frieze--Kannan~\cite{DFK} suffice, together with the reductions
for well-presented bodies in~\cite{DGH,GLS}. One may also use faster
volume implementations, such as Gaussian cooling~\cite{CV} and the
rounding methods of~\cite{JLLV}. These volume query bounds cover only
part of the present algorithm. We also count the simulation of derived
body oracles, the smoothing computation, and the evaluation of the
inverse gradient map.
\section*{Acknowledgements}
This work was supported by the Department of Atomic Energy, Government of
India (project number RTI4014), by the Infosys--Chandrasekharan Virtual
Center for Random Geometry at the Tata Institute of Fundamental Research
(TIFR).

This work was carried out while the author was a Visiting Scientist at
the Simons Institute for the Theory of Computing.


\clearpage
\appendix
\numberwithin{equation}{section}
\section{Two rectangles}
\label{app:rectangle-example}

We apply the representation in Section~\ref{sec:representation} to the
rectangles
\[
  L_1=[-2,2]\times[-1,1],\qquad
  L_2=[-1,1]\times[-3,3],\qquad Q_i=L_i+\tau B_2,
\]
where $\tau>0$. The bodies $Q_i$ have rounded corners. We first compute
the coefficient $b_{(1,1)}^\tau$ of $t_1t_2$ in
$\Vol_2(t_1Q_1+t_2Q_2)$ directly, and then obtain the same coefficient
from~\eqref{eq:coefficient-integral}. The calculation holds for every
$\tau>0$. The restriction on $\tau$ in Section~\ref{sec:balls} is used only
in the approximation estimate.

\subsection{The volume polynomial}

For $t_1,t_2>0$, set
\[
  a_t=2t_1+t_2,\qquad b_t=t_1+3t_2,\qquad
  \rho_t=\tau(t_1+t_2).
\]
By Minkowski addition,
\[
  Q_t=t_1Q_1+t_2Q_2
  =\bigl([-a_t,a_t]\times[-b_t,b_t]\bigr)+\rho_t B_2.
\]
The rectangle has area $4a_tb_t$. The four strips along its edges have total
area $4(a_t+b_t)\rho_t$, and the four quarter-disks have total area $\pi\rho_t^2$.
It follows that
\[
  \Vol_2(Q_t)=4a_tb_t+4(a_t+b_t)\rho_t+\pi\rho_t^2.
\]
Expanding these terms gives
\begin{align*}
  4a_tb_t
    &=8t_1^2+28t_1t_2+12t_2^2,\\
  4(a_t+b_t)\rho_t
    &=12\tau t_1^2+28\tau t_1t_2+16\tau t_2^2,\\
  \pi\rho_t^2
    &=\pi\tau^2t_1^2+2\pi\tau^2t_1t_2+\pi\tau^2t_2^2.
\end{align*}
Thus
\begin{equation}\label{eq:example-geometric-coefficient}
  \boxed{b_{(1,1)}^\tau=28+28\tau+2\pi\tau^2.}
\end{equation}

\subsection{Gaussian smoothing}

For $u=(x,y)$, the two support functions are
\[
  h_1(x,y)=2|x|+|y|,\qquad
  h_2(x,y)=|x|+3|y|.
\]
Let $Z\sim N(0,1)$ and define
\[
  m(s)=\E|s+Z|,\qquad
  \phi(s)=\frac{1}{\sqrt{2\pi}}e^{-s^2/2},\qquad
  \Phi(s)=\int_{-\infty}^s\phi(z)\,dz.
\]
By differentiation,
\[
  m'(s)=\Pr(s+Z>0)-\Pr(s+Z<0)=2\Phi(s)-1,
  \qquad m''(s)=2\phi(s).
\]
Let $\kappa(s)=2\phi(s)$ and $r=\sqrt{1+x^2+y^2}$. The Gaussian
coordinates are standard normal. Thus the functions defined in
Section~\ref{sec:representation} are
\[
  f_1(x,y)=2m(x)+m(y)+\tau r,\qquad
  f_2(x,y)=m(x)+3m(y)+\tau r.
\]

\subsection{The Hessians}

The Hessian of the square-root term, before multiplication by $\tau$, is
\[
  J(x,y):=\nabla^2r
  =\frac1{r^3}
  \begin{pmatrix}
    1+y^2&-xy\\
    -xy&1+x^2
  \end{pmatrix}.
\]
Consequently,
\[
  H_1=\begin{pmatrix}2\kappa(x)&0\\0&\kappa(y)\end{pmatrix}+\tau J,
  \qquad
  H_2=\begin{pmatrix}\kappa(x)&0\\0&3\kappa(y)\end{pmatrix}+\tau J.
\]
The Gaussian terms in this example have positive-definite Hessians at every
point. The additional term gives the curvature bound~\eqref{eq:curvature},
which is needed for general bodies.

Expanding the determinant of the symmetric $2\times2$ matrix gives
\[
  \det(t_1H_1+t_2H_2)
  =t_1^2\det H_1+t_1t_2d_\tau(x,y)+t_2^2\det H_2,
\]
where
\begin{align*}
  d_\tau(x,y)
   &=(H_1)_{11}(H_2)_{22}+(H_2)_{11}(H_1)_{22}
     -2(H_1)_{12}(H_2)_{12}\\
   &=(2\kappa(x)+\tau J_{11})(3\kappa(y)+\tau J_{22})\\
   &\quad +(\kappa(x)+\tau J_{11})(\kappa(y)+\tau J_{22})
     -2\tau^2J_{12}^2.
\end{align*}
Collecting terms, we obtain
\begin{equation}\label{eq:example-density}
\begin{aligned}
  d_\tau(x,y)
   &=7\kappa(x)\kappa(y)
     +\tau\bigl(3\kappa(x)J_{22}+4\kappa(y)J_{11}\bigr)
     +2\tau^2\det J.
\end{aligned}
\end{equation}
Each term is nonnegative. Equation~\eqref{eq:coefficient-integral} identifies
the integral over $\R^2$ with $b_{(1,1)}^\tau$.

\subsection{Integration}

\textbf{The Gaussian terms.}
Since $\int_{\R}\kappa(s)\,ds=2$, we have
\[
  \int_{\R^2}7\kappa(x)\kappa(y)\,dx\,dy
  =7\cdot2\cdot2=28.
\]
This is the mixed coefficient of the two rectangles.

\textbf{The terms linear in the radius.}
For fixed $x$,
\[
  J_{22}(x,y)=\frac{1+x^2}{(1+x^2+y^2)^{3/2}}
  =\frac{\partial}{\partial y}
    \frac{y}{\sqrt{1+x^2+y^2}},
\]
and therefore
\[
  \int_{\R}J_{22}(x,y)\,dy=1-(-1)=2.
\]
Similarly, $\int_{\R}J_{11}(x,y)\,dx=2$ for each fixed $y$. It follows that
\[
  \int_{\R^2}\tau\bigl(3\kappa(x)J_{22}+4\kappa(y)J_{11}\bigr)\,dx\,dy
  =\tau(3\cdot2\cdot2+4\cdot2\cdot2)=28\tau.
\]

\textbf{The quadratic term.}
The determinant of $J$ is
\[
  \det J=\frac{(1+x^2)(1+y^2)-x^2y^2}{r^6}
  =\frac1{(1+x^2+y^2)^2}.
\]
Integrating in polar coordinates, with radial variable $s$, gives
\[
  \int_{\R^2}\det J\,dx\,dy
  =2\pi\int_0^\infty\frac{s}{(1+s^2)^2}\,ds=\pi.
\]
The last term therefore has integral $2\pi\tau^2$. Adding the three integrals gives
\begin{equation}\label{eq:example-integral}
  \boxed{
  \int_{\R^2}[t_1t_2]\det(t_1H_1+t_2H_2)\,dx\,dy
  =28+28\tau+2\pi\tau^2.}
\end{equation}
This is the coefficient in~\eqref{eq:example-geometric-coefficient}.

\subsection{The gradient map and mixed area}

For $F_t=t_1f_1+t_2f_2$, the gradient map is
\[
  \nabla F_t(x,y)
  =\begin{pmatrix}
     a_t(2\Phi(x)-1)\\
     b_t(2\Phi(y)-1)
   \end{pmatrix}
   +\rho_t\frac{(x,y)^{\mathsf T}}{\sqrt{1+x^2+y^2}}.
\]
The first vector lies in $(-a_t,a_t)\times(-b_t,b_t)$ and the second has
norm less than $\rho_t$. Hence their sum belongs to the interior of $Q_t$.
By the diffeomorphism argument in Section~\ref{sec:representation}, every
point in this interior has a unique preimage. The Jacobian is
$t_1H_1+t_2H_2$, so its determinant integrates to the area.

By the definition of mixed volume in dimension two,
\[
  \Vol_2(t_1Q_1+t_2Q_2)
  =t_1^2\Vol_2(Q_1)+2t_1t_2V(Q_1,Q_2)+t_2^2\Vol_2(Q_2).
\]
It follows that
\begin{equation}\label{eq:example-mixed-area}
  \boxed{V(Q_1,Q_2)=14+14\tau+\pi\tau^2.}
\end{equation}
Letting $\tau\downarrow0$ gives $V(L_1,L_2)=14$. The $t_1t_2$ coefficient
for the rectangles is therefore $28$.

\clearpage
\section{Comparison with earlier algorithms}
\label{sec:nr-comparison}

\paragraph{Vertex descriptions.}
Narayanan--Roy~\cite{NR} give a randomized relative approximation
algorithm for lattice polytopes $P_i\subseteq B_\infty(2^L)$, each
specified as the convex hull of at most $m_0$ integer points. Their
running time is polynomial in $n,m_0,L,\widetilde A,\varepsilon^{-1}$,
and $\log\delta^{-1}$, where, for positive multiplicities,
\[
 \widetilde A=\prod_{i=2}^k
       \frac{(\alpha_i+1)^{\alpha_i+1}}{\alpha_i^{\alpha_i}}.
\]
A stronger theorem in that paper uses a refined capacity parameter $A$.
The number of bodies is allowed to vary. For fixed $k$, their parameter
bound gives polynomial running time. Their algorithm first applies capacity
scaling. It then samples the Minkowski sum and recovers the cell of a mixed
subdivision containing the sampled point. The estimator is the indicator
that the dimensions of the summand faces equal the prescribed multiplicities.

\paragraph{Membership oracles.}
In this paper the bodies are rational subspace sections of convex bodies
given by membership oracles, together with inner and outer bounds.
Section~\ref{sec:sections-reduction} detects zero mixed volume exactly and
reduces the positive case to full-dimensional bodies. We then use a smooth
mixed Monge--Amp\`ere representation. It gives a determinant weight in
$[0,1]$ whose expectation is the desired coefficient divided by the volume
of the Minkowski sum. To evaluate the weight, we solve $\nabla F(u)=X$,
estimate the Hessians, and compute their determinant coefficient by exact
rational arithmetic. Support optimization, boundary trimming, and the
regularization of the Hessians give the required finite-bit estimates.
We compare this optimization with the cell-recovery problem of
Narayanan--Roy in Section~\ref{intro:kkt}. Each statistic has the required
expectation. We do not use a pointwise equality between the smooth weight
and the subdivision indicator.

\paragraph{The complexity parameters.}
The bound in~\cite{NR} is already polynomial in $\Gamma_\alpha$ for
vertex-described inputs. To see this, discard zero multiplicities and write
\[
 \widetilde A
 =\prod_{i=2}^k(\alpha_i+1)
                 (1+\alpha_i^{-1})^{\alpha_i}
 \le(2e)^{k-1}\prod_{i=2}^k\alpha_i,
\]
whereas the two-sided Stirling bounds give
\[
 \Gamma_\alpha^2\ge
 (2\pi)^{k-1}e^{-1/(6n)}
       \frac{\prod_{i=1}^k\alpha_i}{n}.
\]
Consequently, in any ordering of the positive multiplicities,
\begin{equation}\label{eq:nr-parameter-comparison}
 \frac{\widetilde A}{\Gamma_\alpha^2}
 \le e^{1/(6n)}\left(\frac e\pi\right)^{k-1}
                       \frac n{\alpha_1}
 \le2n.
\end{equation}
For $k=1$ both parameters equal one. Theorem~\ref{thm:main} extends a
guarantee with polynomial dependence on this ratio to rational sections
of membership-oracle bodies. The proof uses the exact section reduction
and the finite-bit evaluation of the smooth statistic. In particular,
the polynomial running-time bound for fixed $k$ holds for these oracle inputs.

\paragraph{Capacity and repeated bodies.}
Gurvits~\cite{Gurvits} proves a capacity bound with factor $n^n/n!$
for the squarefree coefficient, which equals $n!$ times the conventional
mixed volume. Polarization yields
\[
 p_\alpha\le\Capacity_\alpha(p)\le\Gamma_\alpha p_\alpha,
 \qquad \Capacity_\alpha(p)=\inf_{t>0}p(t)/t^\alpha.
\]
Gurvits computes capacity using convex optimization and volume approximation.
Both sampling algorithms use this normalization. In
Section~\ref{sec:uniform-capacity} we give the finite-precision implementation
and verify that the polynomial degree is independent of the number of bodies.

\paragraph{Interpolation and dominant multiplicities.}
Dyer--Gritzmann--Hufnagel~\cite[Theorems~10--11]{DGH} give randomized relative
approximation for well-presented convex bodies when one multiplicity is at
least $n-\psi(n)$, where $\psi$ is nondecreasing and
$\psi(n)\log\psi(n)=o(\log n)$. Their method interpolates approximate volumes
of Minkowski sums. Their Theorem~4 proves $\#\mathrm P$-hardness of exact computation
even for two rectangular parallelotopes. This does not preclude randomized
approximation. By Remark~\ref{rem:polynomial-regimes}, our algorithm has
randomized polynomial running time in their dominant-multiplicity regime
and when the number of distinct bodies is fixed.

\paragraph{Surface area through heat flow.}
The identity
\[
 \operatorname{area}(\partial K)
 =\left.\frac{d}{ds}\Vol_n(K+sB_n)\right|_{s=0^+}
 =nV(K[n-1],B_n)
\]
places surface area within the present problem. Belkin--Narayanan--Niyogi
\cite{BNN} estimate it through the small-time escape probability
\[
 \Pr\{X+\sqrt{2s}\,G\notin K\}
 =\sqrt{s/\pi}\,
 \frac{\operatorname{area}(\partial K)}{\Vol_n(K)}(1+o(1)),
 \qquad s\downarrow0,
\]
for independent $X\sim\Unif(K)$ and $G\sim N(0,I_n)$. Their finite-time
bounds permit the estimation of this mixed-volume quantity from samples.

\paragraph{Mixed Monge--Amp\`ere measures.}
The identity in Section~\ref{sec:representation} is a smooth case of the
mixed Monge--Amp\`ere identity. Passare--Rullg{\aa}rd
\cite[Proposition~3, pp.~496--497]{PR} identify the total mixed real
Monge--Amp\`ere mass with the mixed volume of bounded domains of convex
conjugates.

For smooth convex functions, the mixed Monge--Amp\`ere densities are mixed
discriminants of the Hessians. We use this fact to evaluate the statistic.

\subsection{Convex duality}
\label{intro:kkt}

We compare the local optimization in Narayanan--Roy~\cite{NR} with the
optimization used to evaluate our weight.

\subsubsection{Subdivision and cell recovery}

After capacity scaling, the target coefficient divided by the volume of the
Minkowski sum is bounded below by the reciprocal of a capacity factor.
Both algorithms estimate this ratio by sampling. Narayanan--Roy~\cite{NR}
use a subdivision of the Minkowski sum of vertex-described lattice polytopes
into sums of faces. Given a sampled point $X$, they solve
\[
 \max\left\{\sum_i\langle a_i,y_i\rangle:
       y_i\in\lambda_iP_i,\quad \sum_i y_i=X\right\},
\]
where the $a_i$ are suitably generic directions. If $F_i$ is the minimal face
of $P_i$ containing $y_i/\lambda_i$, the Karush--Kuhn--Tucker (KKT) conditions
give a common multiplier $v$ with
\[
 a_i-v\in N(\lambda_iF_i,\lambda_iP_i)\qquad(i=1,\ldots,k).
\]
Here $N(F,P)$ denotes the outer normal cone of the face $F$ of $P$.
These normal-cone conditions determine the cell in Schneider's subdivision.
The statistic is
\[
 I_\alpha(X)=\mathbf{1}\{\dim F_i=\alpha_i\text{ for every }i\}.
\]
Genericity and boundary estimates ensure that the correct cell can be
recovered. The KKT argument is in Section~6, Lemma~4, Claim~1 of~\cite{NR}.
It is used to recover the cell; it does not extract a coefficient from a
given determinant polynomial.

\subsubsection{Smooth optimization}

Let $h_K(u)=\sup_{y\in K}\langle u,y\rangle$. The dual of the preceding
maximization is
\[
 \min_{u\in\R^n}
 \left\{\sum_i h_{\lambda_iP_i}(u+a_i)-\langle u,X\rangle\right\}.
\]
Indeed, using $v$ for the multiplier of $X-\sum_i y_i$ gives the dual objective
$\langle v,X\rangle+\sum_i h_{\lambda_iP_i}(a_i-v)$; setting $u=-v$ gives the
displayed expression. Its optimality condition is
\[
 X\in\sum_i\partial h_{\lambda_iP_i}(u+a_i),
\]
where $\partial h$ denotes the convex subdifferential. The faces are
therefore determined by subgradients of the support functions.

After the section reduction, we work with full-dimensional bodies. In
Sections~\ref{sec:balls} and~\ref{sec:representation} we construct the
enlarged bodies $Q_i$ and smooth convex functions $f_i$ satisfying
$H_i=\nabla^2 f_i\succ0$. For $Q=\sum_iQ_i$ and $F=\sum_i f_i$, we find
the inverse image of $X\in\interior Q$ by solving
\[
 u_* = \mathop{\rm argmin}_{u\in\R^n}
       \{F(u)-\langle u,X\rangle\},
 \qquad \sum_i\nabla f_i(u_*)=X.
\]
This has the same form as the preceding dual problem. Here the functions
are smooth, and their gradients replace the subgradient face data. The
Hessians give the Jacobian of the change of variables.

We then evaluate
\[
 W_\alpha(X)=
 \frac{[z^\alpha]\det(\sum_i z_iH_i(u_*))}
      {\det(\sum_i H_i(u_*))}.
\]
Section~\ref{sec:estimator} proves the bound on this weight and computes
its expectation. Lemma~\ref{lem:numerical-evaluation} gives the approximation
to the inverse point and the Hessians. For the resulting rational matrices,
Section~\ref{sec:determinant-arithmetic} computes the coefficient exactly.

\end{document}